\documentclass[conference]{IEEEtran}
\IEEEoverridecommandlockouts
\usepackage{cite}
\usepackage{amsthm}
\newtheorem{proposition}{Proposition}
\usepackage{amsmath,amssymb,amsfonts}
\usepackage{algorithmic}
\usepackage{graphicx}
\usepackage{textcomp}
\usepackage{pgfplotstable,filecontents}
\usepackage{circuitikz}
\usetikzlibrary{fit}
\usepackage{svg}
\usepackage{url}
\usepackage{nicefrac}
\usepackage{csquotes}
\usepackage{tabularx}
\usepackage{booktabs}
\usepackage{makecell}
\usepackage{braket}
\usepackage{balance}
\usepackage{yquant}
\yquantset{register/default name=$\ket{\reg_{\idx}}$}

\usepackage{algorithm2e}
\usepackage{bm}
\usepackage[inline]{enumitem}
\usepackage{graphbox, graphicx}
\usepackage{hyperref}
\usepackage{mathtools}
\usepackage{pgfplots}
\usepackage{subcaption}

\renewcommand{\vec}[1]{\boldsymbol{#1}}
\newcommand{\mat}[1]{\boldsymbol{#1}}
\newcommand{\trn}[1]{#1^{\top}}
\newcommand{\ipt}[2]{\trn{#1}{#2}}

\newcommand{\bb}{\boldsymbol{b}}

\newcommand{\bC}{\boldsymbol{C}}

\newcommand{\cU}{\mathcal{U}}

\DeclareMathOperator*{\argmax}{arg\,max}
\DeclareMathOperator*{\argmin}{arg\,min}
\newcommand{\amin}[1]{\argmin_{#1}}
    
\begin{document}

\bstctlcite{IEEEexample:BSTcontrol}

\title{Towards Bundle Adjustment for Satellite Imaging\\via Quantum Machine Learning}

\makeatletter
\newcommand{\linebreakand}{%
  \end{@IEEEauthorhalign}
  \hfill\mbox{}\par
  \mbox{}\hfill\begin{@IEEEauthorhalign}
}
\makeatother

\author{
\IEEEauthorblockN{Nico Piatkowski}
\IEEEauthorblockA{{Fraunhofer IAIS, ME Group}\\
nico.piatkowski@iais.fraunhofer.de}
\and
\IEEEauthorblockN{Thore Gerlach}
\IEEEauthorblockA{{Fraunhofer IAIS, ME Group}\\
thore.gerlach@iais.fraunhofer.de}
\and
\IEEEauthorblockN{Romain Hugues}
\IEEEauthorblockA{{Thales Alenia Space}\\
romain.hugues@thalesaleniaspace.com}
\linebreakand
\IEEEauthorblockN{Rafet Sifa}
\IEEEauthorblockA{{Fraunhofer IAIS, ME Group}\\
rafet.sifa@iais.fraunhofer.de}
\and
\IEEEauthorblockN{Christian Bauckhage}
\IEEEauthorblockA{{Fraunhofer IAIS, ME Group}\\
christian.bauckhage@iais.fraunhofer.de}
\and
\IEEEauthorblockN{Frederic Barbaresco}
\IEEEauthorblockA{{Thales Land and Air Systems}\\
frederic.barbaresco@thalesgroup.com}
}

\maketitle

\begin{abstract}
Given is a set of images, where all images show views of the same area at different points in time and from different viewpoints. The task is the alignment of all images such that relevant information, e.g., poses, changes, and terrain, can be extracted from the fused image. In this work, we focus on quantum methods for keypoint extraction and feature matching, due to the demanding computational complexity of these sub-tasks. To this end, $k$-medoids clustering, kernel density clustering, nearest neighbor search, and kernel methods are investigated and it is explained how these methods can be re-formulated for quantum annealers and gate-based quantum computers. Experimental results obtained on digital quantum emulation hardware, quantum annealers, and quantum gate computers show that classical systems still deliver superior results. However, the proposed methods are ready for the current and upcoming generations of quantum computing devices which have the potential to outperform classical systems in the near future. 
\end{abstract}


\section{Introduction}
Quantum computing devices became available recently. It is, however, important to understand that these devices underlie heavy resource limitations. First, the number of available qubits limits the sheer size of the problems that can be solved. Second, deep circuits are subject to decoherence, which destroys the quantum state. Third, the quantum processors are subject to various sources of noise that affects the computation. Thus, various algorithms which enjoy superior theoretical properties, e.g. amplitude amplification \cite{Brassard2000QuantumAA}, cannot be executed faithfully by the current and upcoming generations of noisy intermediate scale quantum (NISQ) computers.
Despite these limitations, we explain how a relevant class of computer vision problems can be transferred to the quantum domain. More precisely, we consider a situation in which a set of $3$-dimensional points from a scene is viewed by $m$ cameras. Given a list of image coordinates of these points in the camera coordinates, finding the set of camera positions, altitudes, imaging parameters, and the point's $3$-dimensional locations is a reconstruction process. Bundle adjustment is the estimation that involves the minimization of the re-projection error. It usually goes through an iterative process and requires a good initialization \cite{Hartley/Zisserman/2003a}.
Solutions to this task allow for the extraction of $3$-dimensional coordinates describing the image geometry and the intrinsic coordinate system of each image. By fusing the available images to construct a single large image, one may eventually extract valuable information, e.g. a classification of each pixel of whether it belongs to a moving object.
One may split this process into various sub-tasks: 
The first step consists of \emph{extracting keypoints} (also called interest points or feature points) from every image. These correspond to characteristic pixels of the images, e.g. corner points. In the naive setting, all pixels may act as a keypoint. %
Next, keypoints which are common to multiple images must be \emph{matched}. As the set of pixels can be large, these two tasks can be very computationally demanding. Based on the correspondence between feature points, one may finally identify a projection that aligns the coordinate systems of all images. When all images are fused, missing areas are identified and overlapping areas with contradicting pixel data are segmented to identify moving objects.\\
\par The final alignment step is done by finding transformations between different images which align them to a single plane. These transformations can have varying forms, e.g. a homography or a fundamental matrix. 
Without further processing, this sub-tasks is inherently continuous and thus not very well suited for quantum computation. Classical methods include the \enquote{eight-point algorithm} \cite{abdel2015direct}, direct linear transformation (DLT) \cite{hartley1997defense}, and enhanced correlation coefficient (ECC) \cite{evangelidis2013efficient}.
The refinement of the alignment step is referred to as \emph{bundle adjustment} \cite{triggs1999bundle}.\\
\par Mathematically, the problem can be formulated as follows: Assume that $n$ $3$-dimensional points are visible through $m$ different views and let $\vec{p}_{ij}$ be the projection of the $i$-th point onto the plane containing the $j$-th image. Since $\vec{p}_{ij}$ may not lie in the image itself, we define a binary variable $b_{ij}$ which is 1 if and only if point $i$ is visible in image $j$. Furthermore, assume that the camera that created the $j$-th image can be characterized by a vector $\vec{w}_j$, and every $3$-dimensional point $i$ by a vector $\vec{r}_i$. The objective is now to minimize the total re-projection error
\begin{align}
f_{BA}(\vec{w}_j, \vec{r}_i)=\sum_{i=1}^{n}\sum_{j=1}^{m}b_{ij}\left\|\pi(\vec{w}_j, \vec{r}_i)- \vec{p}_{ij}\right\|^2_2,\label{eq:ba_objective}
\end{align}
where $\pi$ corresponds to the predicted projection. 

The goal of this paper is to investigate how far current quantum computing resources can be used to tackle the problem of bundle adjustment with machine learning. %
As explained above, we focus on keypoint extraction and feature matching, which are both computationally hard problems. Due to their discrete structure, these sub-tasks exhibit a large potential for improvements via quantum computation, as known from other areas of signal processing \cite{Presles/etal/2021a}.\\
\par Our contributions can be summarized as follows:
\begin{itemize}
	\item We propose a novel re-interpretation of keypoint extraction as a clustering problem. Based on this insight, we introduce two NISQ-compatible quantum keypoint extraction methods (Sec.~\ref{sec:quantum_clustering}). 
	\item Given keypoints, we construct a Hamiltonian whose ground state is the solution to the matching problem (Sec.~\ref{sec:quantum_matching}). Again, resulting in a NISQ-compatible quantum algorithm. 
	\item To the best of our knowledge, we propose the first combination of adiabatic quantum computing and quantum gate computing, whereas a gate-based quantum kernel function is used within a quadratic unconstrained binary optimization (QUBO) problem that is solved on a quantum annealer.
	\item An experimental evaluation of our methods on digital quantum emulators and actual quantum hardware shows the effectiveness and NISQ-compatibility of our methods (Sec.~\ref{sec:experiments}).
\end{itemize}

\section{Notation and Background}
Let us summarize the notation and background necessary for the subsequent development. The set that contains the first $n$ strictly positive integers is denoted by $[n]=\{1,2,\dots,n\}$ and $[a,b]=[b]\setminus [a-1]$ with $[0]=\emptyset$. 

\subsection{Gate-based Quantum Computing}
An $N$-qubit quantum circuit $\bC$ takes an input state $\ket{\psi_{in}}$---typically the all-$0$ state $\ket{0^N}$---and generates an $N$-bit vector $\psi_{out}\in\{0,1\}^N$. 
The act of reading out the result $\bb$ from $\bC$ is called {\em measurement}. Here, any quantum state vector $\ket{\psi}$ denotes a $2^N$-dimensional complex vector. State vectors are always normalized such that $\sum_{i=1}^{2^N} |\ket{\psi}_i|^2 = 1$. Moreover, the squared absolute value of the $i$-th dimension of $\ket{\psi}$ is the probability for measuring the binary representation of $i$ as output of circuit $\bC$. E.g., the probability for measuring the binary representation of $10$ is $|\braket{10\mid\psi}|^2$, where $\braket{a\mid b}$ denotes the ordinary inner product between vectors $\ket{a}$ and $\ket{b}$. 
All gate-based operations on the $N$ input qubits must be unitary operators, typically acting on one or two qubits at a time. These low-order operations can be composed to form more complicated qubit transformations. 
Any unitary operator $U$ satisfies $U^\dagger U=U U^\dagger=I$ and its eigenvalues have modulus (absolute value) $1$. Here, $I$ denotes the identity and $\dagger$ denotes the conjugate transpose. 
Borrowing terminology from digital computing, unitary operators acting on qubits are also called \emph{quantum gates}. 
In the context of this work, we will be especially interested in the gates\[
  \sigma_x = \begin{bmatrix}
0 & 1\\
1 & 0
\end{bmatrix}
,\hspace{0.25cm}
  \sigma_z = \begin{bmatrix}
1 & 0\\
0 & -1
\end{bmatrix}
,\hspace{0.25cm}
  \sigma_I = \begin{bmatrix}
1 & 0\\
0 & 1
\end{bmatrix},
\]and $H=\frac{1}{\sqrt{2}}(\sigma_x+\sigma_z)$. The matrices $\sigma_x$, $\sigma_I$, and $\sigma_z$ are called Pauli matrices. 
Finally, the action of any quantum circuit $\bC$ can be written as a product of unitaries:
\[
\bC = U_d U_{d-1} U_{d-2} \dots U_1\;,
\]
where $d$ is the {\em depth} of the circuit. An exemplary quantum gate circuit is shown in Fig.~\ref{fig:qkernelmachine}. 
It is important to understand that a quantum gate computer receives its circuit symbolically as a sequence of low dimensional unitaries---the implied $2^n \times 2^n$ matrix is never materialized. 
A detailed introduction into that topic can be found in \cite{Nielsen/Chuang/2016a}. 

\subsection{Adiabatic Quantum Computing}
Adiabatic quantum computing (AQC) relies on on the \emph{adiabatic theorem} \cite{Born1928-BDA} which states that, if a quantum system starts in the \emph{ground state} of a \emph{Hamiltonian operator} which then gradually changes over a period of time, the system will end up in the ground state of the resulting Hamiltonian. Since Hamiltonians are energy operators, their ground states correspond to the lowest energy state of the quantum system under consideration. 

To harness AQC for problem solving, one prepares a system of qubits in the ground state of a simple, problem independent Hamiltonian and then adiabatically evolves it towards a Hamiltonian whose ground state corresponds to a solution to the problem at hand \cite{Albash2018-AQC}. This can be done on quantum annealers \cite{Bian2010-TIM,Johnson2011-QAW,Lanting2014-EIA,Henriet/etal/2020a} which are particularly tailored towards solving \emph{quadratic unconstrained binary optimization} problems of the form
\begin{equation}
\label{eq:AQC:QUBO}
\vec{z}^* = \amin{\vec{z} \in \{0,1\}^N} F(\vec{z}) = \amin{\vec{z} \in \{0,1\}^N} \trn{\vec{z}} \mat{Q} \, \vec{z} + \ipt{\vec{q}}{\vec{z}}.
\end{equation}

The connection between QUBOs and quantum computing becomes evident by considering the following Hamiltonian:
\begin{equation}\label{eq:AQC:QUBOH}
H_{\mat{Q}} = \sum_{i=1}^N \sum_{j=1}^N \, \mat{Q}_{ij} \, \frac{\sigma_I^i-\sigma_z^i}{2} \, \frac{\sigma_I^j-\sigma_z^j}{2} + \sum_{i=1}^N \, \vec{q}_i \, \frac{\sigma_I^i-\sigma_z^i}{2}\;,
\end{equation}
where  $\sigma_z^j = I \otimes  \ldots \otimes I \otimes \sigma_z \otimes I \otimes  \ldots \otimes I$ denotes the Pauli matrix $\sigma_z$ acting on the $j$-th qubit. By design, we have a 1-to-1 correspondence between Eqs.~\eqref{eq:AQC:QUBO} and \eqref{eq:AQC:QUBOH} such that the smallest eigenvalue of $H_{\mat{Q}}$ is identical to the minimum of $F$. Moreover, the eigenstate of $H_{\mat{Q}}$ that corresponds to the minimum eigenvalue is the pure state $\ket{\vec{z^*}}$. 

QUBOs are not only of interest for quantum annealers: 
The same type of problem can be solved on gate-based quantum computers via quantum approximate optimization \cite{farhi2014quantum,hadfield2019quantum} also known as variational quantum eigensolver \cite{Peruzzo2014-AVE}. Recent results suggest that these techniques can be more robust against exponentially small spectral gaps, and thus, they may deliver better solutions than actual quantum annealers \cite{Zhou2020-QAO}. 

Whenever the dimension of a QUBO exceeds the number of available hardware qubits, we consider a splitting procedure that is described together with the experimental details.

\section{Methodology}\label{sec:methods}
In this section we describe our methods for approaching the tasks of keypoint extraction and feature matching.
In general, pixel data can either be represented by raw color channels, or sophisticated feature space mappings, e.g., scale-invariant feature transform (SIFT) \cite{lowe2004distinctive}, (accelerated) KAZE \cite{alcantarilla2012kaze,alcantarilla2011fast}, or low-dimensional embeddings based on geometric hashing \cite{lang2010astrometry,evangelidis2013efficient}. 
Nevertheless, if not stated differently, 
we assume that an image is encoded as a set of pixels $I=\{\vec{p}_1,\dots,\vec{p}_{M}\}$, where $\vec{p_i}=(x_i,y_i,r_i,g_i,b_i)$ represents a pixel with location $(x_i,y_i)$ and color channels $(r_i,g_i,b_i)$. 

\subsection{Quantum Keypoint Extraction}\label{sec:quantum_clustering}
The goal of keypoint detection is to extract a subset of relevant pixels which describe the full image well\textemdash{}we re-interpreted this step as clustering problem. Due to the NP-hardness of clustering, offloading the corresponding computation to a quantum processor promises large benefits. %

\subsubsection{Quantum $k$-Medoids Clustering}
For solving the $k$-medoids clustering problem on a quantum computer, we reformulate the classical $k$-medoids objective into a QUBO problem \cite{bauckhage2019qubo}.
The idea is to combine the selection of $k$ mutually far apart objects as well as the selection of $k$ most central objects:
\begin{align*}
\argmax_{C\subset \mathcal I}\sum_{i,j=1}^{k} \left\|\vec{c}_i- \vec{c}_j\right\|_2,\quad
\argmin_{C\subset \mathcal I}\sum_{i=1}^{k}\sum_{j=1}^{M} \left\|\vec{c}_i- \vec{p}_j\right\|_2,
\end{align*}
where $\mathcal I\supset C=\{\vec{c}_1,\dots,\vec{c}_{k}\}$ corresponds to the set of $k$ cluster medoids.

By considering the distance matrix $\mat{D}$ with $D_{ij}=\left\|\vec{p}_i- \vec{p}_j\right\|_2$, these two optimization objectives can put together into one QUBO problem.
\begin{proposition}\label{prop:k_medoids}
A QUBO formulation for the $k$-Medoids clustering problem 
is given by
\begin{align}
\mat{Q}=\gamma \mathbf{1}_N\trn{\mathbf{1}_N}-\alpha \mat{D},\ \vec{q}=\beta \mat{D}\mathbf{1}_N-2\gamma k\mathbf{1}_N,
\label{eq:qubo_k_medoid}
\end{align}
\end{proposition}
with dimension $N=M$, where $\alpha$, $\beta$ and $\gamma$ are Lagrange multipliers and $\mathbf{1}_N$ is the $N$-dimensional vector consisting only of ones.
\begin{proof}
For a detailed derivation, we refer to \cite{bauckhage2019qubo}.
\end{proof}

\subsubsection{Quantum Kernel Density Clustering}\label{sec:kdc}
 
Kernel density clustering (KDC) maintains estimates of the probability densities of pixels via kernel density (Parzen) estimates
\begin{align}
\rho_{\mathcal I}(\vec{p})=\frac{1}{M}\sum_{i=1}^{M}K(\vec{p}, \vec{p}_i), \quad
\rho_C(\vec{p})=\frac{1}{k}\sum_{j=1}^{k}K(\vec{p}, \vec{c}_j), \label{eq:parzen}
\end{align}
where $K(\cdot,\cdot)$ is a kernel function. We assume $K$ to be a Mercer kernel, i.e. there exists some feature map $\phi$, such that
\begin{align*}
K(\cdot,\cdot)=\braket{\phi(\cdot),\ \phi(\cdot)}.
\end{align*}
The density estimates in \eqref{eq:parzen} can then be rewritten as
\begin{align}\rho_X(\vec{p})=\braket{\phi(\vec{p}),\ \hat\phi_X},\ \hat\phi_X=\frac{1}{|X|}\sum_{\vec{x}\in X}\phi(\vec{x}),\ X\subset \mathcal I.
\label{eq:kdc_feauture_map}
\end{align}
The KDC problem can be formulated as minimizing the discrepancy between the feature map distributions in \eqref{eq:kdc_feauture_map} by finding optimal cluster centroids $C$:
\begin{align*}
\argmin_C\left\|\hat\phi_{\mathcal I}-\hat\phi_C\right\|_2^2
=\argmin_C
\braket{\hat\phi_C,\ \hat\phi_C} - 2 \braket{\hat\phi_{\mathcal I},\ \hat\phi_C}. 
\end{align*} 
\begin{proposition}\label{prop:kdc}
A QUBO formulation for the KDC problem is given by
\begin{align}
\mat{Q}=\frac{1}{k^2}\mat{\mathcal K} +\lambda\mathbf{1}_N\trn{\mathbf{1}_N},\ \vec{q}=-2\left( \frac{1}{kN}\mat{\mathcal K}\mathbf{1}_N+\lambda k \mathbf{1}_N\right), \label{eq:qubo_kdc}
\end{align}
with dimension $N=M$, where $\lambda$ is a Lagrange multiplier and $\mat{\mathcal K}$ is the kernel matrix with $\mathcal K_{ij}=K(\vec{p}_i,\vec{p}_j)$.
\end{proposition}
\begin{proof}
For a detailed derivation we refer to \cite{bauckhage2020hopfield}.
\end{proof}
Up to now, classical Gaussian kernels, with $K(\boldsymbol{x},\boldsymbol{y})=\exp(-\gamma\|\boldsymbol{x}-\boldsymbol{y}\|_2^2)$, have been considered for KDC in the literature. Here, however, we also consider quantum kernels, i.e. the kernel matrix $\mat{\mathcal K}$ is computed via the quantum gate circuit that is shown in Fig.~\ref{fig:qkernelmachine}. 

In contrast to $k$-means or $k$-medoids, the kernel density estimation results in cluster centroids that contain more information for dense parts of an image and hence deliver substantially different results.

\subsection{Quantum Feature Matching}\label{sec:quantum_matching}

Based on the output of keypoint extraction, suiting matches between keypoints of different images must be identified. To this end, let $\vec{p}^{(1)}_i, 1\leq i\leq n$ and $\vec{p}^{(2)}_j, 1\leq j\leq m$ be keypoints extracted from images (1) and (2), respectively. The task is to find pairs $(\vec{p}^{(1)}_{i}, \vec{p}^{(2)}_{j})$ s.t. $\vec{p}^{(1)}_i$ in the first image, corresponds to $\vec{p}^{(2)}_{j}$ in the second image.

For this task, raw information about a single pixel is insufficient, since both images might be scaled, rotated or illuminated in a different way. Thus, feature descriptors, e.g. SIFT or AKAZE, are computed for each keypoint, i.e. $\vec{x}_i$ for $\vec{p}^{(1)}_i$ and $\vec{y}_j$ for $\vec{p}^{(2)}_j$.
In addition to the bare descriptors, Kernel methods can improve the matching further by projecting the descriptors into a very high-dimensional feature space. %

\subsubsection{Feature Matching as QUBO}\label{sec:matching_qubo}
Let $\mat{\mathcal K}$ be the matrix 
$\left(K(\vec{x}_i,\vec{y}_j) \right)_{i,j=1}^{n,m}$ of kernel
values between all pairs of keypoints, $\mat{\mathcal K}\in\mathbb{R}^{n\times m}$.
\begin{proposition}\label{prop:matching}
A QUBO formulation for the matching problem is given by
\begin{align}
\mat{Q}&=\sum_{i=1}^{n}\beta \mat{A}_i+\gamma\left( \vec{b}_i\trn{\vec{b}_i}+\vec{c}_i\trn{\vec{c}_i}+\vec{b}_i\trn{\vec{c}_i}+\vec{c}_i\trn{\vec{b}_i}\right), \label{eq:matching_qubo_matrix}\\
\vec{q}&= \trn{\mat{P}}\left((1-\alpha)\mathbf{1}_{nm}-\text{vec}\left(\mat{\mathcal K}\right) \right)-2\gamma k\sum_{i=1}^{n}\trn{\vec{b}_i}\vec{c}_i,
\label{eq:matching_qubo_offset}
\end{align}
with dimension $N=n(m+l)$, $\mat{A}_i=\trn{\mat{P}}\mat{S}_i\mat{P}$, $\vec{b}_i=\trn{\mat{P}}\mathbf{1}_{nm}^{(i)}$ and $\vec{c}=\trn{\vec{d}}\mat{P}_i$, where
\begin{align*}
&\mat{P}\in \{0,1\}^{nm\times N},\ \mat{P}\vec{z}=\vec{z}_{[1,nm]}, 
\\
&\mat{P}_i\in\{0,1\}^{l\times N},\ \mat{P}_i\vec{z}=\vec{z}_{[nm+(i-1)l+1,nm+il]},
\end{align*}
are projection matrices acting on $\vec{z}\in \{0,1\}^{N}$. The parameter $k$ corresponds to the maximum number of matches for an $\vec{x}_i$,  $l=\left\lceil\log_2(k+1)\right\rceil$ and 
the parameters $\alpha$, $\beta$ and $\gamma$ are Lagrange multipliers. Furthermore for some $\vec{v}\in\{0,1\}^{nm}$
\begin{align*}
&\vec{d}\in\mathbb N^{l},\ \vec{d}=\trn{(2^0,2^1,\dots,2^{l-1})}, \\
& \mathbf{1}_{mn}^{(i)}\in\{0,1\}^{nm},\ \mathbf{1}_{mn}^{(i)}=\trn{(0, \dots, 0, \overbrace{1,\dots, 1}^{[(i-1)m+1,im]},0,\dots 0)}, \\
&\mat{S}_i\vec{v}=\trn{( \vec{v}^{(i)}, \dots, \vec{v}^{(i)},\overbrace{0,\dots, 0}^{[(i-1)m+1:im]}, \vec{v}^{(i)}, \dots, \vec{v}^{(i)})}, 
\end{align*}
with $\mat{S}_i\in\{0,1\}^{nm\times nm}, \vec{v}^{(i)}=\vec{v}_{[(i-1)m+1,im]}$.
\end{proposition}
\begin{proof}
The optimization objective for the matching problem can be formulated as follows:
\begin{align}
\min_{\vec{v}\in\{0,1\}^{nm}}& \quad \trn{\vec{v}}\left(\mathbf{1}_{nm}-\text{vec}\left(\mat{\mathcal K}\right) \right) 
\label{eq:matching_qubo_min}\\
\text{and} \ \min_{\vec{v}\in\{0,1\}^{nm}}& \quad -\trn{\vec{v}}\mathbf{1}_{nm} \label{eq:matching_qubo_max}\\
\text{subject to}& \quad \trn{\vec{v}}\mat{S_i}\vec{v}=0, \ 1\le i\le n, \label{eq:matching_qubo_duplicate} \\
& \quad \trn{\vec{v}}\mathbf{1}_{nm}^{(i)}\le k,\ 1\le i\le n. \label{eq:matching_qubo_neighbors}
\end{align}
The condition in \eqref{eq:matching_qubo_duplicate} ensures that two keypoints $\vec{x}_{i_1}$ and $\vec{x}_{i_2}$ are not matched to the same keypoint $\vec{y}_j$ while \eqref{eq:matching_qubo_neighbors} forces every keypoint $\vec{x}_i$ to be matched with maximally $k$ points $\vec{y}_{j_1},\dots,\vec{y}_{j_{k}}$. 
The inequality constraint in \eqref{eq:matching_qubo_neighbors} can be reformulated by using binary slack variables $\vec{s}_i\in\{0,1\}^{l}$
\begin{align*}
\trn{\vec{v}}\mathbf{1}_{nm}^{(i)}\le k \ &\Leftrightarrow\ \trn{\vec{v}}\mathbf{1}_{nm}^{(i)}+\trn{\vec{d}}\vec{s}_i= k,
\end{align*}
since $\trn{\vec{v}}\mathbf{1}_{nm}^{(i)}\ge 0$ and $\trn{\vec{d}}\vec{s}_i\in [0,k]$. With the technique of Lagrange multipliers the QUBO can then be formulated:
\begin{align*}
&\min_{\vec{v},\vec{s}_1,\dots,\vec{s}_{n}} \trn{\vec{v}}\left((1-\alpha)\mathbf{1}_{nm}-\text{vec}\left(\mat{\mathcal K}\right)\right) \\
&+\beta\sum_{i=1}^{n}\trn{\vec{v}}\mat{S}_i\vec{v} 
+\gamma\sum_{i=1}^{n}\left(\trn{\vec{v}}\mathbf{1}_{nm}^{(i)}+\trn{\vec{d}}\vec{s}_i- k\right)^2 \\
\Leftrightarrow& \min_{\vec{z}\in\{0,1\}^{N}} \trn{(\mat{P}\vec{z})}\left((1-\alpha)\mathbf{1}_{nm}-\text{vec}\left(\mat{\mathcal K}\right)\right) \\
+\beta&\sum_{i=1}^{n}\trn{(\mat{P}\vec{z})}\mat{S}_i\mat{P}\vec{z} 
+\gamma\sum_{i=1}^{n}\left(\trn{(\mat{P}\vec{z})}\mathbf{1}_{nm}^{(i)}+\trn{\vec{d}}\mat{P}_i\vec{z}- k\right)^2,
\end{align*}
since $\mat{P}\vec{z}=\vec{v}$ and $\mat{P}_i\vec{z}=\vec{s}_i$.
\end{proof}
The parameters $\beta$ and $\gamma$ are chosen to be large enough such that the conditions in \eqref{eq:matching_qubo_duplicate} and \eqref{eq:matching_qubo_neighbors} are adhered.
Without loss of generality, we assume $K(\vec{x}_i, \vec{y}_j)\le 1$ and choose $\alpha$ to be in $[0,1]$. A large $\alpha$ emphasizes the maximization of the number of matches in \eqref{eq:matching_qubo_max} which forces every $\vec{x}_i$ to be matched with a $\vec{y}_j$ even though they may not be very similar. 
By setting $\alpha$ close to 0, the distance minimization in \eqref{eq:matching_qubo_min} is prioritized---leading to no matches at all in the extreme case.

\subsubsection{Quantum Kernel Methods}\label{sec:quantum_kernel}

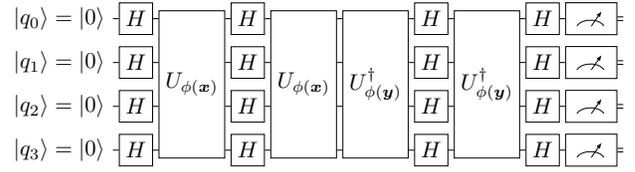
\begin{figure}[!t] %
\centering
\scalebox{0.85}{
		\begin{tikzpicture}
		\begin{yquant}
        qubit {$\ket{q_{\idx}} = \ket0$} q[4];
        h q[0];
        h q[1];
        h q[2];
        h q[3];
        box {$U_{\phi(\vec{x})}$} (q[0], q[1], q[2], q[3]);
        h q[0];
        h q[1];
        h q[2];
        h q[3];
        box {$U_{\phi(\vec{x})}$} (q[0], q[1], q[2],q[3]);
        box {$U^\dagger_{\phi(\vec{y})}$} (q[0], q[1], q[2],q[3]);
        h q[0];
        h q[1];
        h q[2];
        h q[3];
        box {$U^\dagger_{\phi(\vec{y})}$} (q[0], q[1], q[2],q[3]);
        h q[0];
        h q[1];
        h q[2];
        h q[3];
		measure q[0];
		measure q[1];
		measure q[2];
		measure q[3];
		\end{yquant}
		\end{tikzpicture}
		}
\caption{Gate-based 4-qubit quantum kernel. Data points $\boldsymbol{x}$ and $\boldsymbol{y}$ are both mapped into the $2^4$-dimensional Hilbert space via parameters of unitary gates  $U_{\phi(\vec{x})}$ and  $U_{\phi(\vec{y})}$.}
\label{fig:qkernelmachine}
\end{figure}
The feature map of quantum kernel machines is hardware specific. Instead of relying on problem specific knowledge to construct the feature map or the kernel, the intrinsically $2^N$-dimensional Hilbert space of an $N$-qubit register is utilized to realize the feature map \cite{Havlicek/etal/2019a}. A schematic representation of a corresponding quantum gate circuit is shown in Fig.~\ref{fig:qkernelmachine}. There, a maximal superposition is prepared and then passed through $N$-qubit unitaries which create the feature space transformation of the data. It is important to understand that data does not enter the circuit in the discrete qubit state space. Instead, data is passed in form of parameters of universal unitary gates $U_{\phi(\vec{x})}$ and $U_{\phi(\vec{y})}$ representing (parts of) the high-dimensional feature map. Each classical $N$-bit binary string is interpreted as one feature and the corresponding probability amplitudes of the qubit state as feature values. The actual kernel value $K(\vec{x},\vec{y})$ is then given by estimating the transition amplitude $|\left\langle \phi(\vec{x}), \phi(\vec{y}) \right\rangle|^2 = | \bra{0^N} \cU^\dagger_{\phi(\vec{y})} \cU_{\phi(\vec{x})} \ket{0^N}|^2$ with $\cU_{\phi(\cdot)}=U_{\phi(\cdot)}H^{\otimes n}U_{\phi(\cdot)}H^{\otimes n}$. Clearly, the specific choice of $U_{\phi(\cdot)}$ is not fixed and can be tuned for the application at hand. In \cite{Havlicek/etal/2019a}, the authors suggest the following unitary:
\begin{equation}\label{eq:qphi}
U_{\phi(\vec{x})} = \exp\left(-i \sum_{S\subset [N]} \phi_s(\vec{x}) \prod_{v\in S} \sigma_z^v\right).
\end{equation}

Obtaining the full kernel matrix for $N$ data points requires $N (N+1)/2$ runs of the circuit in Fig.~\ref{fig:qkernelmachine}. The resulting quantum kernel matrix is then ready to be used in our density based quantum keypoint extractor.

For compatibility with NISQ-devices, some typical pitfalls must be avoided: 
Considering local feature functions for all subsets ${S\subset [N]}$ is too costly: To see this, one has to consider the transpilation of user specified quantum circuits. Transpilation is the process of rewriting a given input circuit to match the connectivity structure and noise properties of a specific quantum processor. %
Most importantly, it encompasses the decomposition of gates involving three or more qubits into 2-qubit gates. %
As a direct result, an apparently \enquote{shallow} quantum gate circuit, consisting of a single unitary operation among $N$-qubits, can thus eventually exhibit a very high depth. High circuits depths require large decoherence and dissipation times, which are not available in the current generation of NISQ-devices. It is hence recommended to consider only pairwise features in \eqref{eq:qphi}, e.g., $S\subset [N]\wedge|S|=2$.  %

\subsubsection{Quantum Nearest Neighbor Search}\label{sec:nearest_neighbor}

Finding the nearest neighbors of each keypoint via a quantum gate circuits constitutes a possible alternative to our matching QUBO. 
In a recent contribution, Basheer et al.~\cite{Basheer2020-QKN} integrate ideas from \cite{Wiebe2015-QAF} and \cite{Lloyd2013-QML} with techniques for preparing data into amplitudes of quantum states \cite{Mitarai2018-QAD} and present specific quantum circuits for $k$-nearest neighbor search. 
First, the information on the similarity between keypoints is encoded in the amplitudes of quantum states. The similarity between a generic quantum state $\psi$ and the $j$-th keypoint $\phi_j$ is defined via $G_j=\braket{\psi | \phi_j}^2$. This computation is conducted in the sub-circuit $\mathcal C^{\operatorname{amp}}$. Then the digital conversion algorithm is applied to obtain the state $\ket{j}\ket{G_j}$, which is done by the sub-circuit $\mathcal C^{\operatorname{dig}}$. Classical gates are applied to two different states $\ket{j}\ket{G_j}$ and $\ket{y}\ket{G_y}$ to obtain the state $\ket{j}\ket{g_{y,A}(j)}$, where $y\in A$ and $A$ is the set of the current best $k$ candidate nearest neighbors. Finally, the circuit $\mathcal C^g$ computes the Boolean function $g_{y,A}$, which is defined to be
\begin{align*}
g_{y,A}(j)=
\begin{cases}
1, & G_j > G_y \text{ and } j\notin A, \\
0, & \text{else}.
\end{cases}
\end{align*}
The state $\ket{j}\ket{g_{y,A}(j)}$ is then combined with the quantum search algorithm \cite{Gover1996-AFQ}, for finding the maximum. %

In addition to the fidelity that was used as similarity measure in the original work, one can consider the Hamming distance for the use with binary feature descriptors. A corresponding quantum circuit can be found in \cite{ruan2017quantum}, using the incrementation circuit. This incrementation circuit along with the swap test needed for computing $G_j$ is provided in Fig.~\ref{fig:hamming}. 
\begin{figure}[!t] %
	\centering
		\scalebox{0.95}{
			\begin{tikzpicture}
			\begin{yquant}
			qubit {$\ket{q_{\idx}}$} q[3];
			qubit {$\ket{1}$} anc;
			cnot q[0] | anc;
			cnot anc | q[0];
			cnot q[1] | anc;
			cnot anc | q[1] ~ q[0];
			cnot q[2] | anc;
			cnot anc ~ q[0-1];
			\end{yquant}
			\end{tikzpicture}
}
		\scalebox{0.95}{
			\centering
			\begin{tikzpicture}
			\begin{yquant}
			qubit {$\ket{0}$} anc;
			qubit {$\ket{x}$} x;
			qubit {$\ket{y}$} y;
			h anc;
			swap (x, y) | anc;
			h anc;
			measure anc;
			\end{yquant}
			\end{tikzpicture}
		}
	\caption{(Left) Circuit realizing the incrementation operation. In this example we consider an integer $q$ represented in the binary basis $(q_0,q_1,q_3)$. The output of this circuit corresponds to the increment $q+1$ and can be generalized for an arbitrary basis $(q_0,\dots,q_{n-1})$. (Right) Swap test for distance computing. The probability of measuring the auxillary qubit to be $\ket{0}$ is $P(\ket{0}_{a})=\frac{1}{2}+\frac{1}{2}\braket{x|y}^2$ and thus the fidelity can be estimated.}
	\label{fig:hamming}
\end{figure}
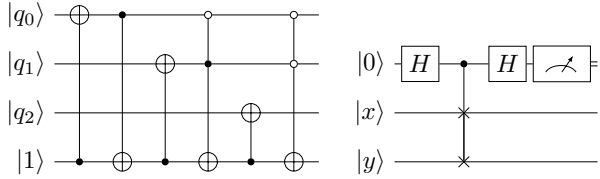

Since NISQ-devices noisy and have a short decoherence time, amplitude amplification techniques like Grover-search cannot implemented faithfully. Nevertheless, upcoming generations of quantum computing hardware will certainly allow us to realize the quantum nearest neighbor-based matching approach.  %

\section{Experimental Evaluation}\label{sec:experiments}

For our experimental evaluation, we consider $344$ sets of $k=5$ images from the Kaggle \enquote{Draper Satellite Image Chronology} challenge\footnote{\url{https://www.kaggle.com/c/draper-satellite-image-chronology/data}} which all have a resolution of $3099\times 2329$ pixels.

\subsection{Experimental Protocol}

Each pixel $\vec{p}$ is represented by a 5-dimensional vector which captures the position in the image as well as the RGB color channels, $\vec{p}=(x,y,r,g,b)$. %
We further down-weight the location information by a factor of $\nicefrac{1}{4}$ to emphasize the importance of the color channels. Finally, pixel vectors are normalized, i.e. $\left\| \vec{p}\right\|_2=1$. %

The raw image resolution implies that the QUBOs from Props.~\ref{prop:k_medoids} and \ref{prop:kdc} are $7\,217\,571$-dimensional---far beyond the capabilities of any quantum annealer or gate-based quantum computer. 
We hence split the task recursively into sub-tasks. First, redundant information is reduced by down-sampling images to $928\times 704$ pixels. Then, each image is split into $32\times 32$ equally sized non-overlapping sub-images, which results in patches of size $29\times 22$ pixels. Fig.~\ref{fig:keypoints_pipeline:fourth} depicts one exemplary patch of the image from Fig.~\ref{fig:keypoints_pipeline:first}. The keypoint extraction on the original image is an iterative process of finding keypoints in the current \enquote{layer} and then merging them to form the next dataset for clustering. %
The Lagrange parameters are chosen such that every summand in \eqref{eq:qubo_k_medoid} and \eqref{eq:qubo_kdc} has approximately the same contribution. 
For $k$-medoids clustering we set $\alpha=1/k$, $\beta=1/n$ and $\gamma=1 /k$, and for KDC we set $\lambda=1/k^2$. Since the parameters $\gamma$ and $\lambda$ weigh the constraint of finding exactly $k$ cluster centroids, setting these values too low can result in finding states not adhering this constraint. %
For the matching problem we employ SIFT feature descriptors in a normalized inner product kernel. The QUBO parameters (see Prop. \ref{prop:matching}) are set to $k=1$, $\beta=1$, $\gamma=1$, while $\alpha$ is varied for showing the effect of this parameter. %

QUBO solvers process the same problem multiple times to prevent local optima. The state with the lowest energy is then chosen to be the solution. We consider a digital annealer (10 shots, runtime $1s$ per shot), a D-Wave Advantage System 5.1 with 5619 qubits (1024 shots, runtime $50\,\mu s$ per shot), simulated annealing, and tabu search, the latter being implemented in the D-Wave Ocean SDK\footnote{\href{https://docs.ocean.dwavesys.com/en/stable/}{https://docs.ocean.dwavesys.com/en/stable/}}. 

For computing quantum kernels, %
we consider a statevector simulation and an IBM Falcon superconducting quantum processor. The circuits on the IBM system are executed with $10\,000$ shots.

\begin{figure*}[t!]
	\centering
	\begin{subfigure}[t]{0.24\textwidth}
		\centering
		\includegraphics[width=\textwidth]{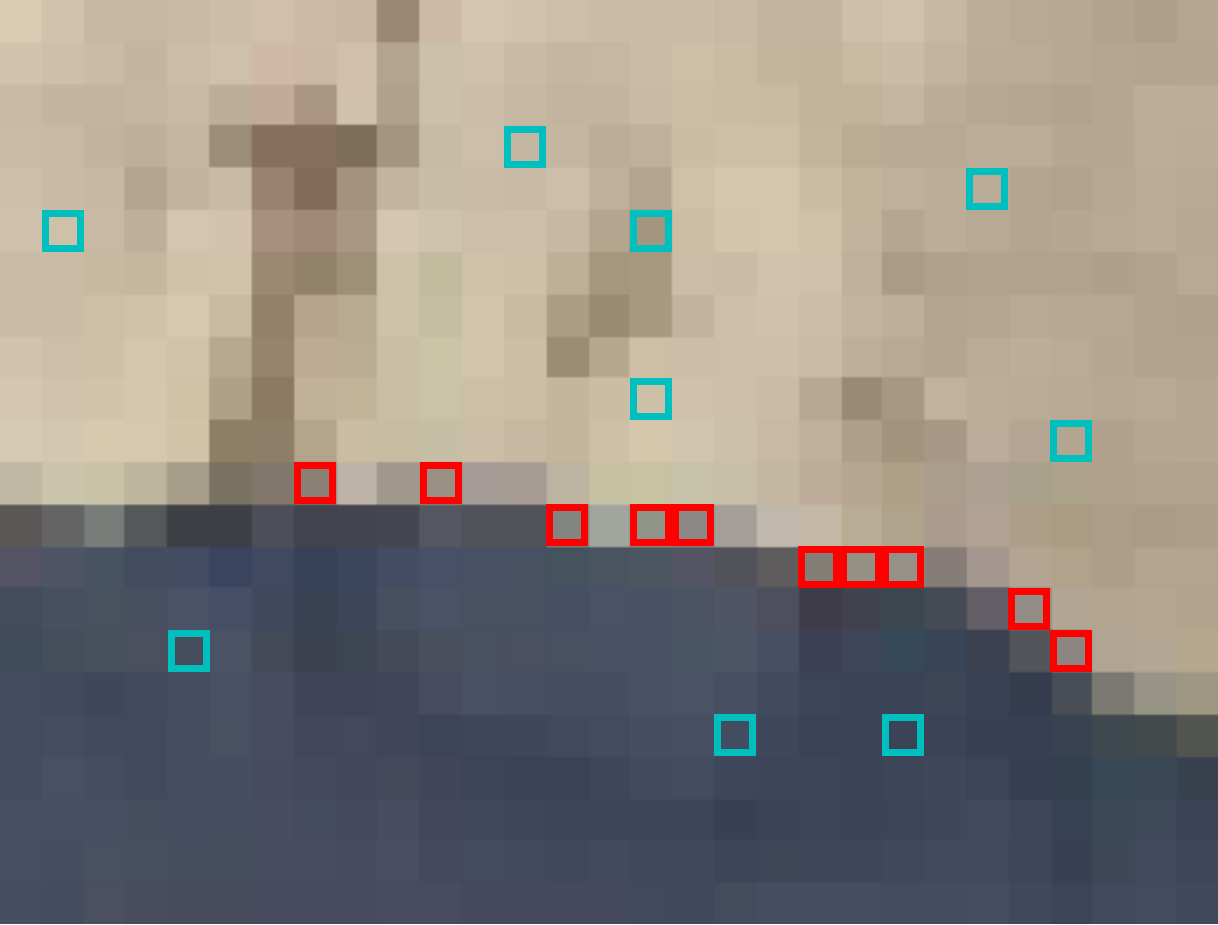}
		\\[0.5cm] \includegraphics[width=\textwidth]{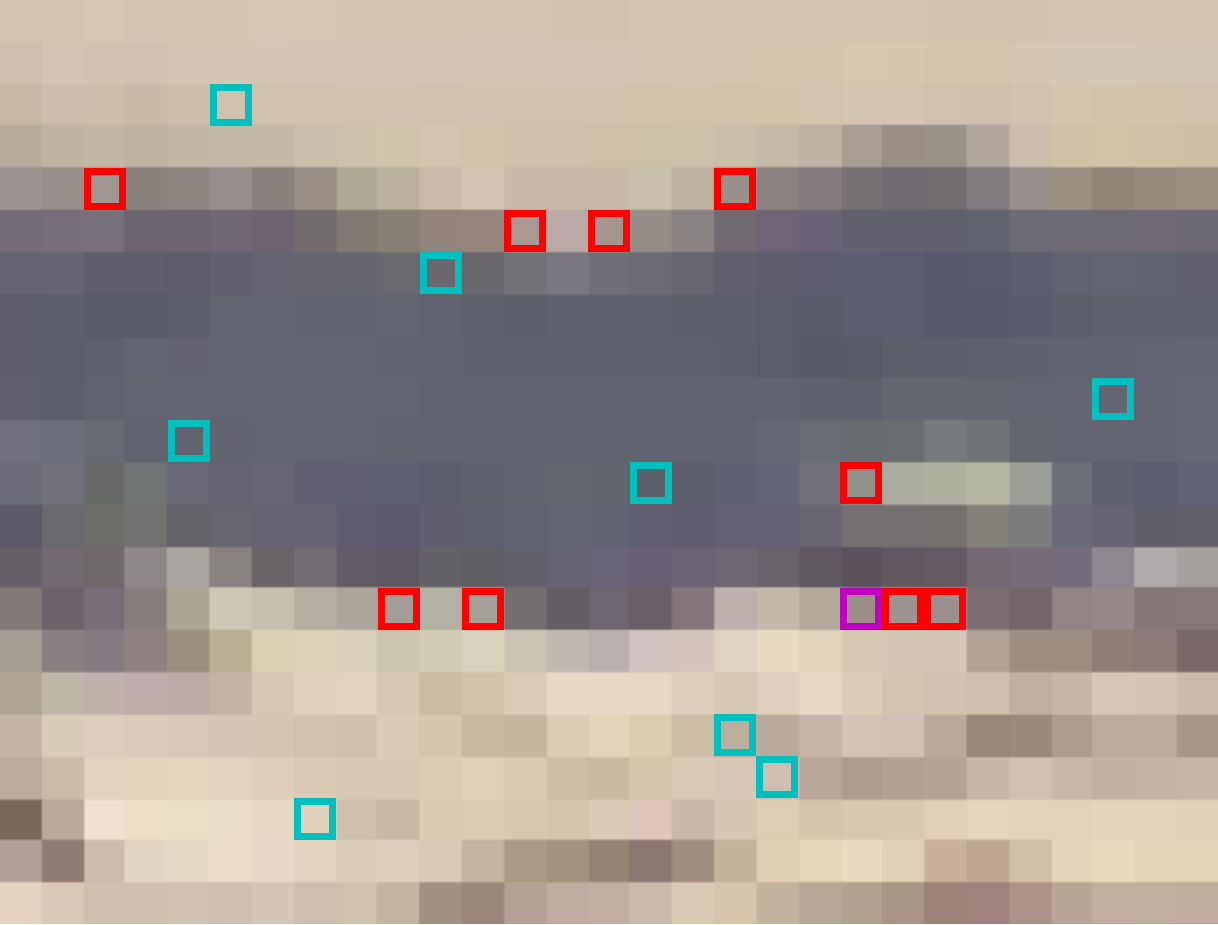}
		\caption{Extraction of 10 keypoints on a single image patch.}
		\label{fig:keypoints_pipeline:fourth}
	\end{subfigure}
	\begin{subfigure}[t]{0.24\textwidth}
		\centering
		\includegraphics[width=\textwidth]{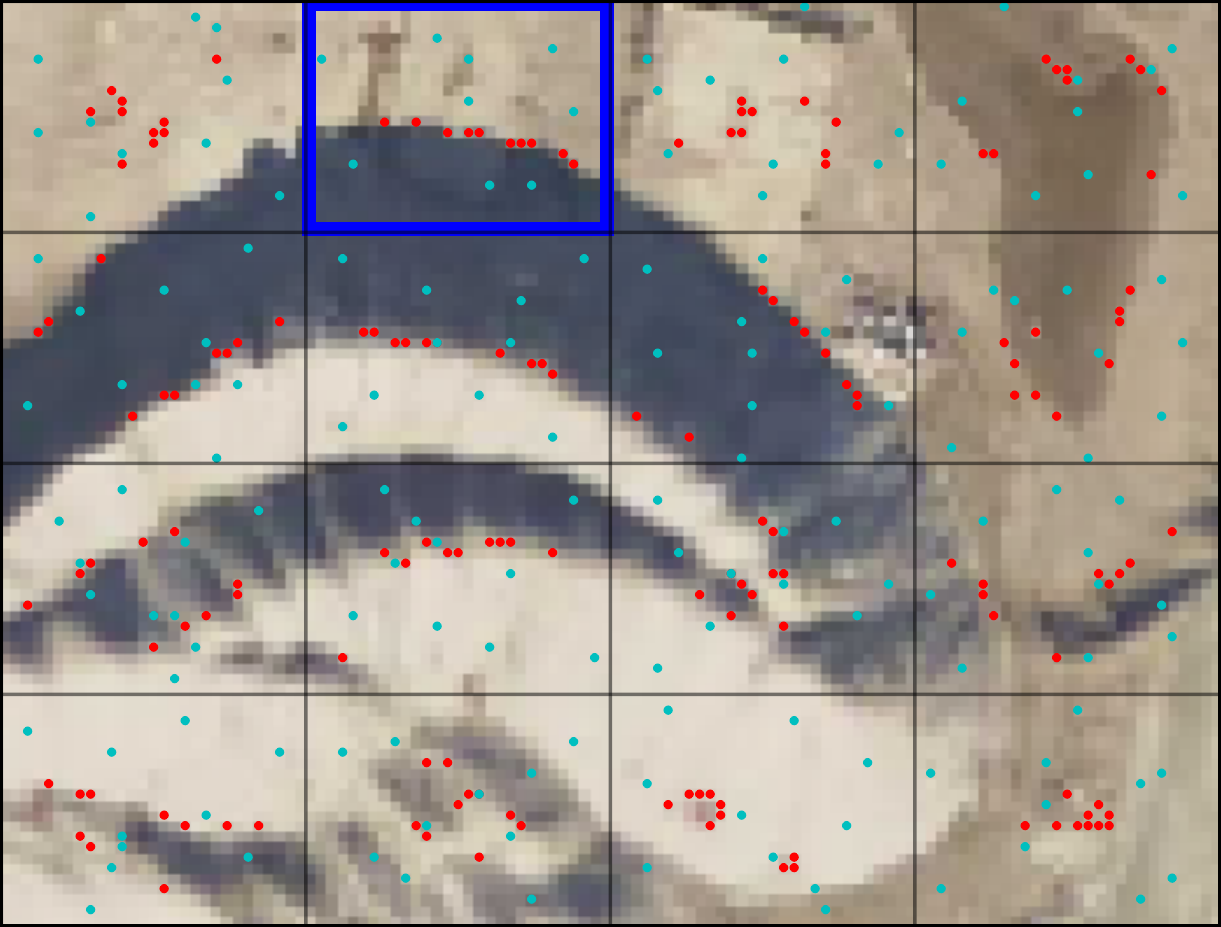}
		\\[0.5cm] \includegraphics[width=\textwidth]{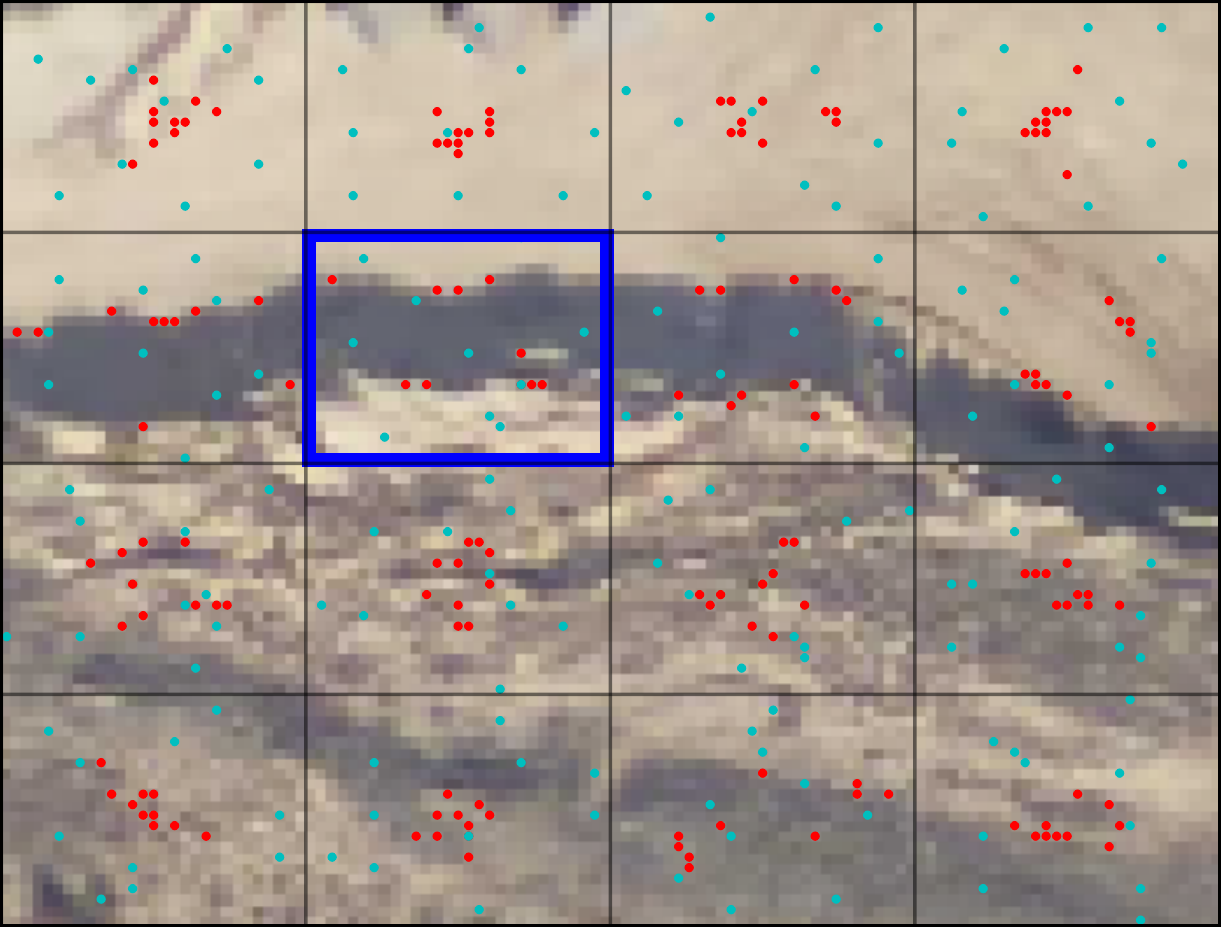}
		\caption{Extraction of 10 keypoints on every image patch.}
		\label{fig:keypoints_pipeline:third}
	\end{subfigure}
	\begin{subfigure}[t]{0.24\textwidth}
		\centering
		\includegraphics[width=\textwidth]{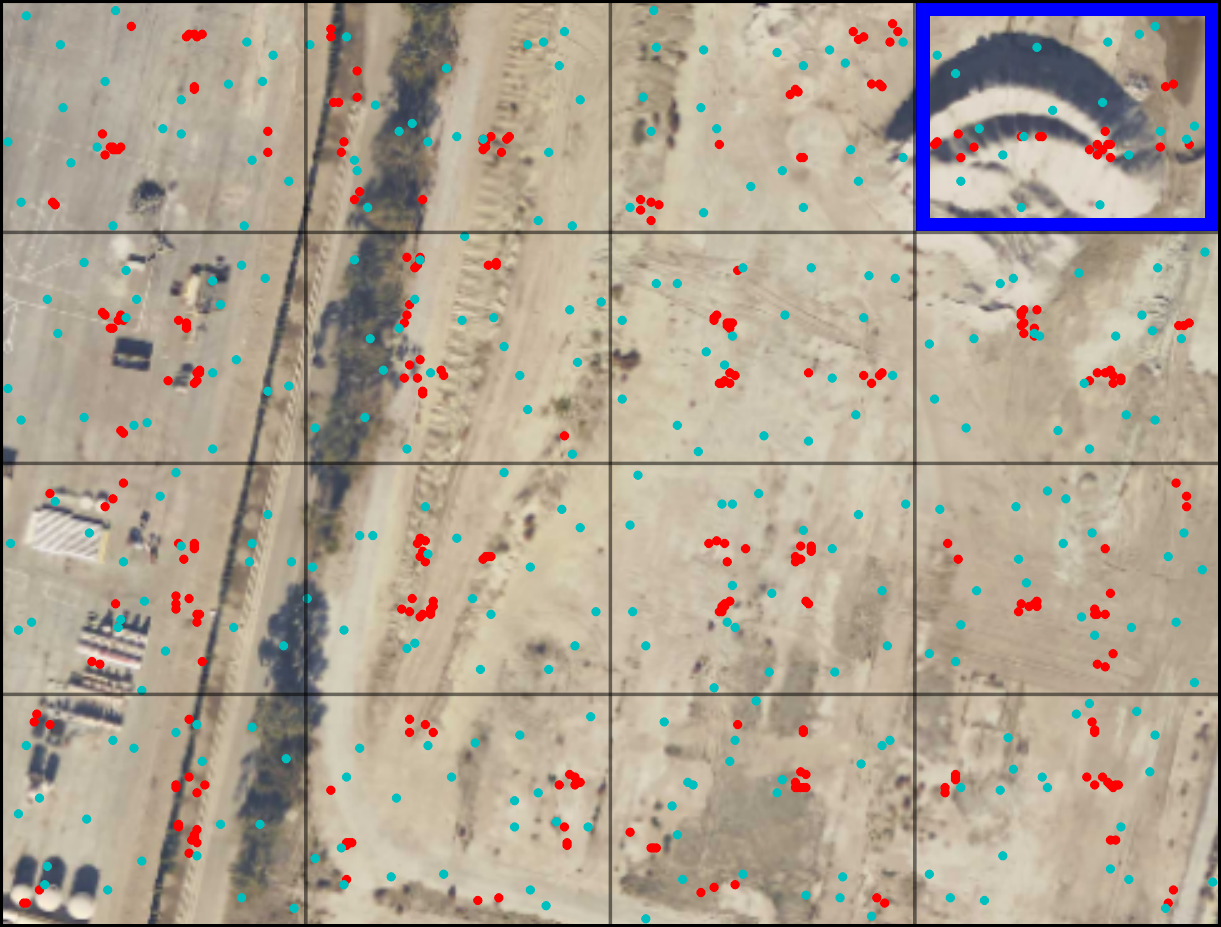}
		\\[0.5cm] \includegraphics[width=\textwidth]{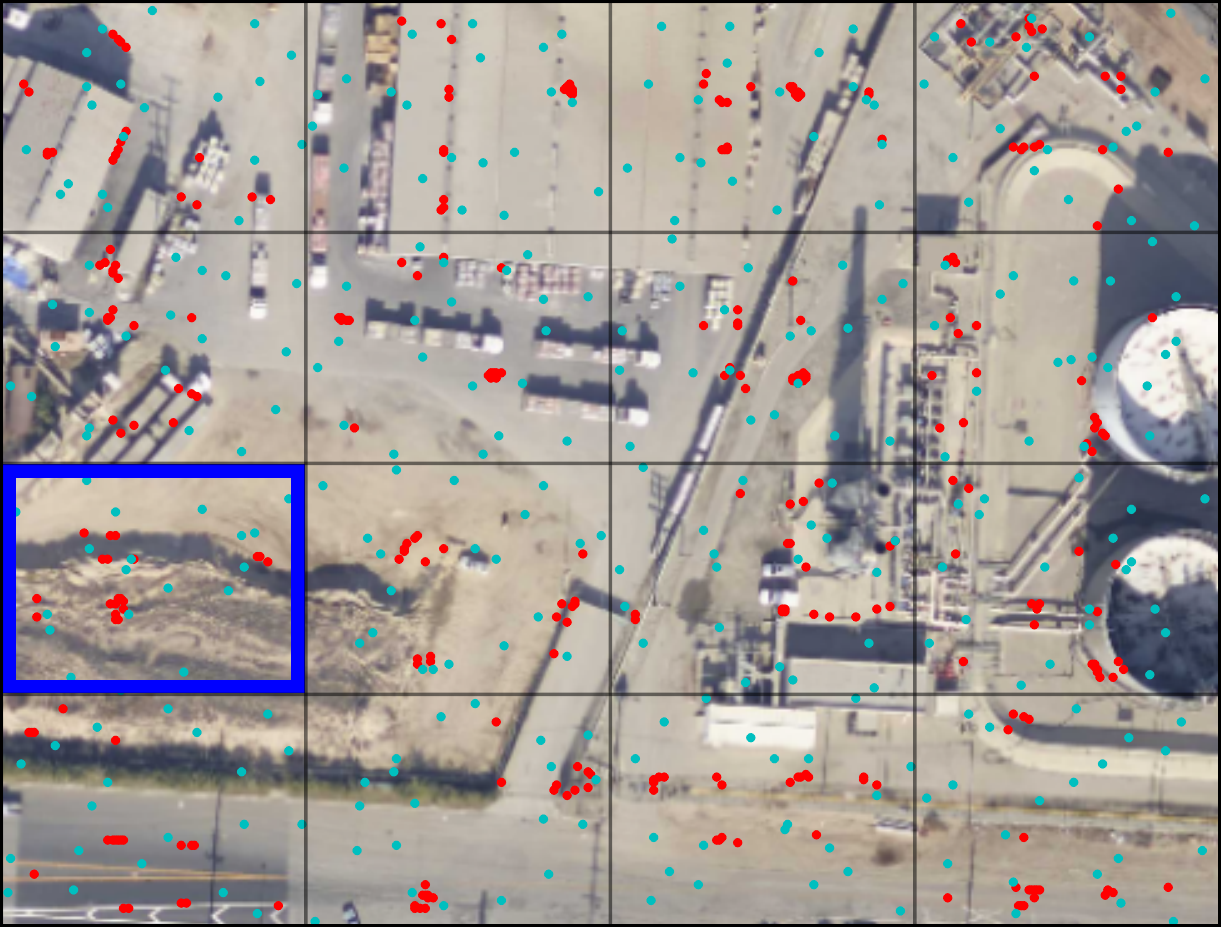}
		\caption{Extraction of 20 keypoints on $4\times4$ grid.}
		\label{fig:keypoints_pipeline:second}
	\end{subfigure}
	\begin{subfigure}[t]{0.24\textwidth}
		\centering
		\includegraphics[width=\textwidth]{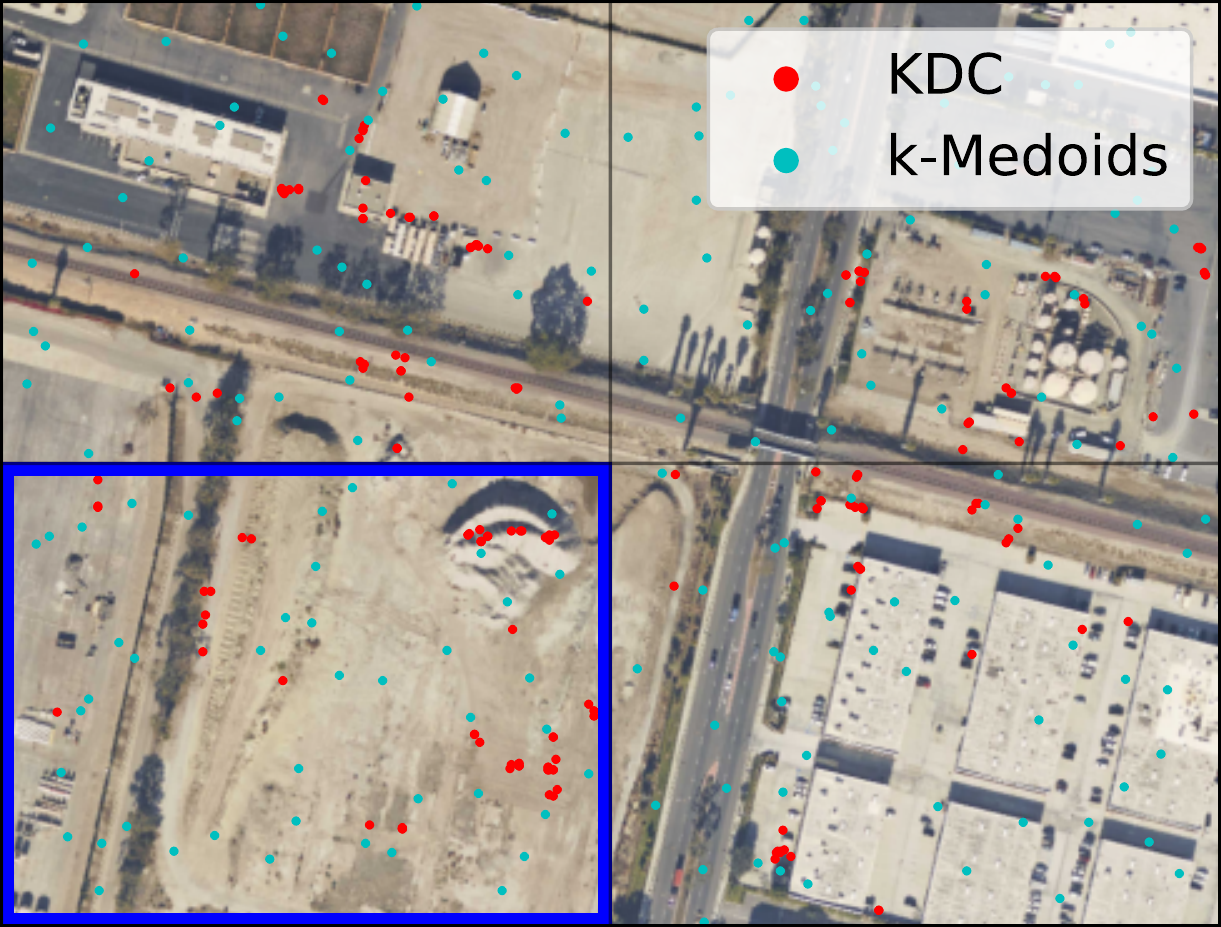}
		\\[0.5cm] \includegraphics[width=\textwidth]{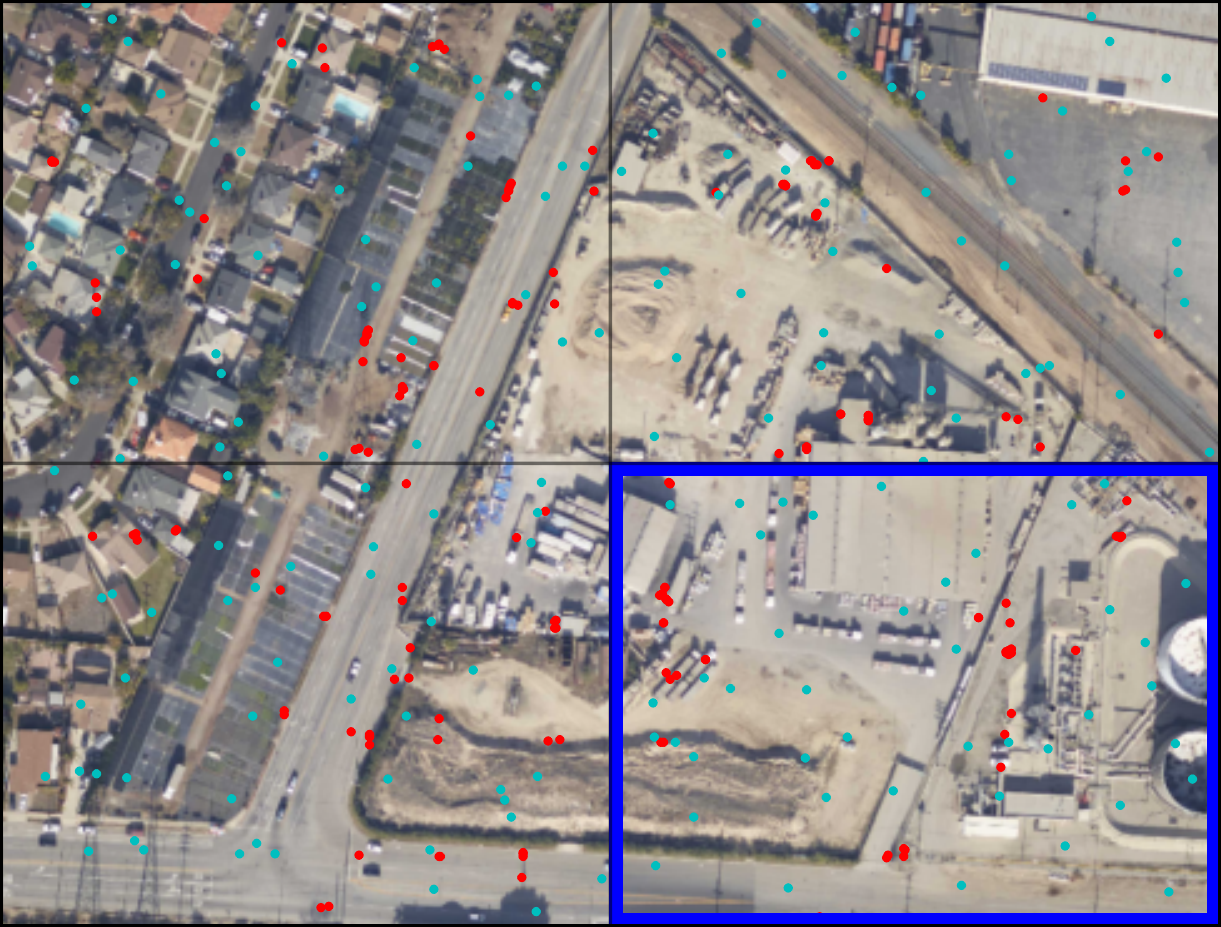}
		\caption{Extraction of 45 (final) keypoints on $4\times4$ grid.}
		\label{fig:keypoints_pipeline:first}
	\end{subfigure}
	\caption{Keypoint extraction is done by subsequently computing cluster centroids on sub-images by solving the QUBOs given in Sec.~\ref{sec:quantum_clustering} using the digital annealers. The image is first split into $32\times32$ image patches of size $29\times22$ pixels and ten keypoints are extracted on every patch. The resulting centroids are then grouped into $8\times 8$ grids of size $4\times 4$ to obtain keypoint sets of size 160 shown in (a) and (b). We extract 20 keypoints on each of these sets and group the resulting cluster centroids into $2\times 2$ grids of size $4\times4$ to obtain keypoint sets of size 320 (c). Lastly, 45 cluster centroids are computed on these sets to obtain the final set of 180 keypoints (d).}
	\label{fig:keypoints_pipeline}
\end{figure*}

\begin{figure*}[t!]
	\centering
	\begin{subfigure}[t]{0.24\textwidth}
		\centering
		\includegraphics[width=\textwidth]{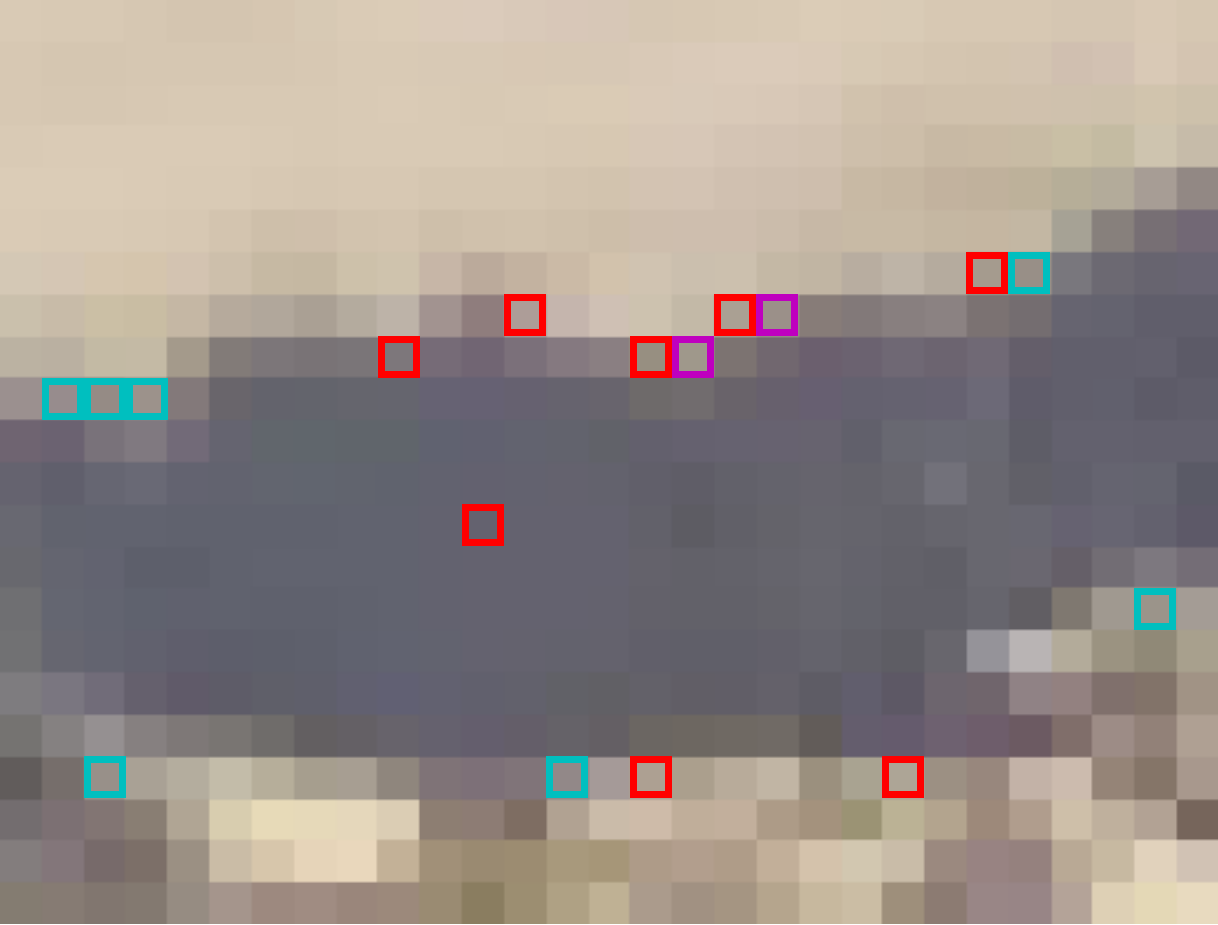}
		\caption{}
		\label{fig:keypoints:first}
	\end{subfigure}
	\begin{subfigure}[t]{0.24\textwidth}
		\centering
		\includegraphics[width=\textwidth]{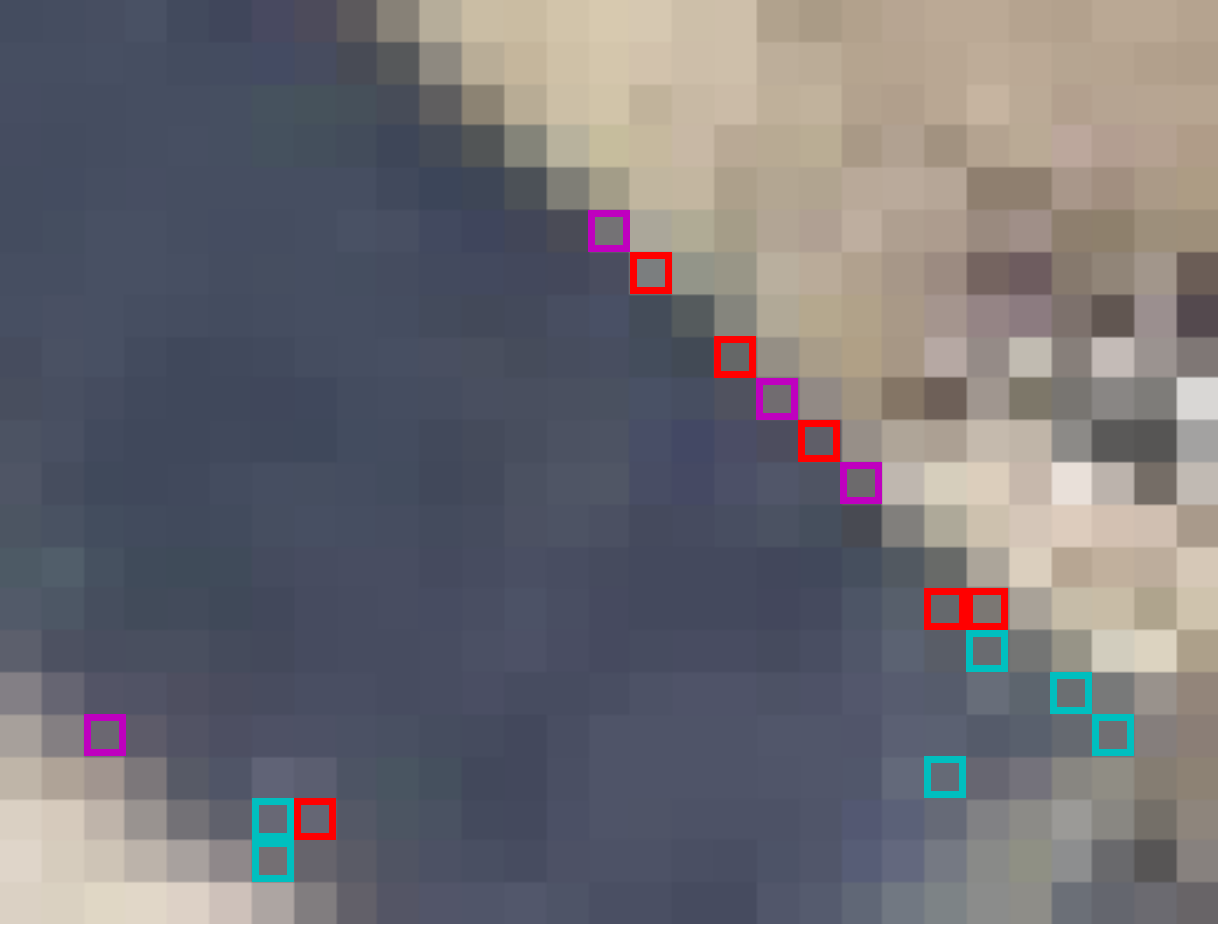}
		\caption{}
		\label{fig:keypoints:second}
	\end{subfigure}
	\begin{subfigure}[t]{0.24\textwidth}
		\centering
		\includegraphics[width=\textwidth]{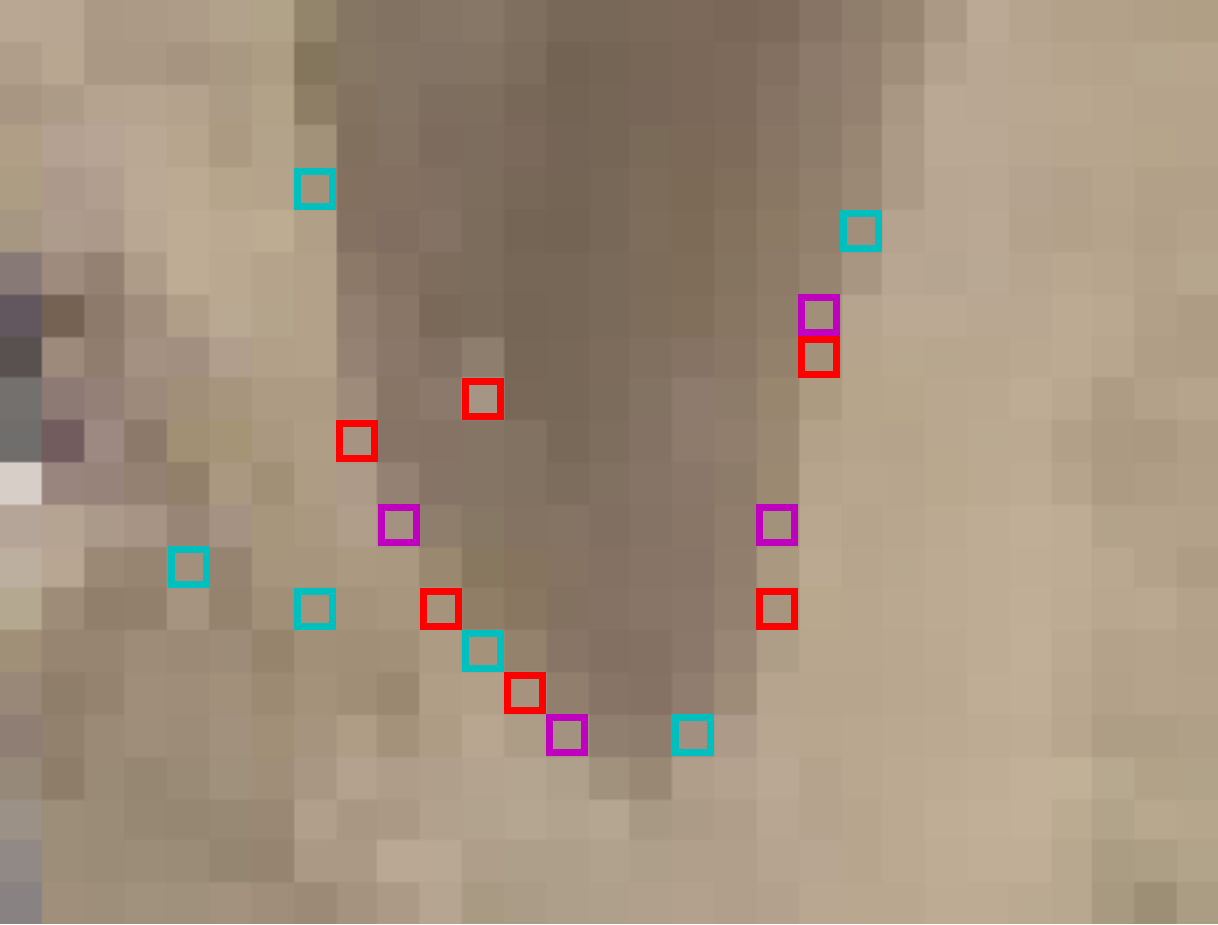}
		\caption{}
		\label{fig:keypoints:third}
	\end{subfigure}
	\begin{subfigure}[t]{0.24\textwidth}
		\centering
		\includegraphics[width=\textwidth]{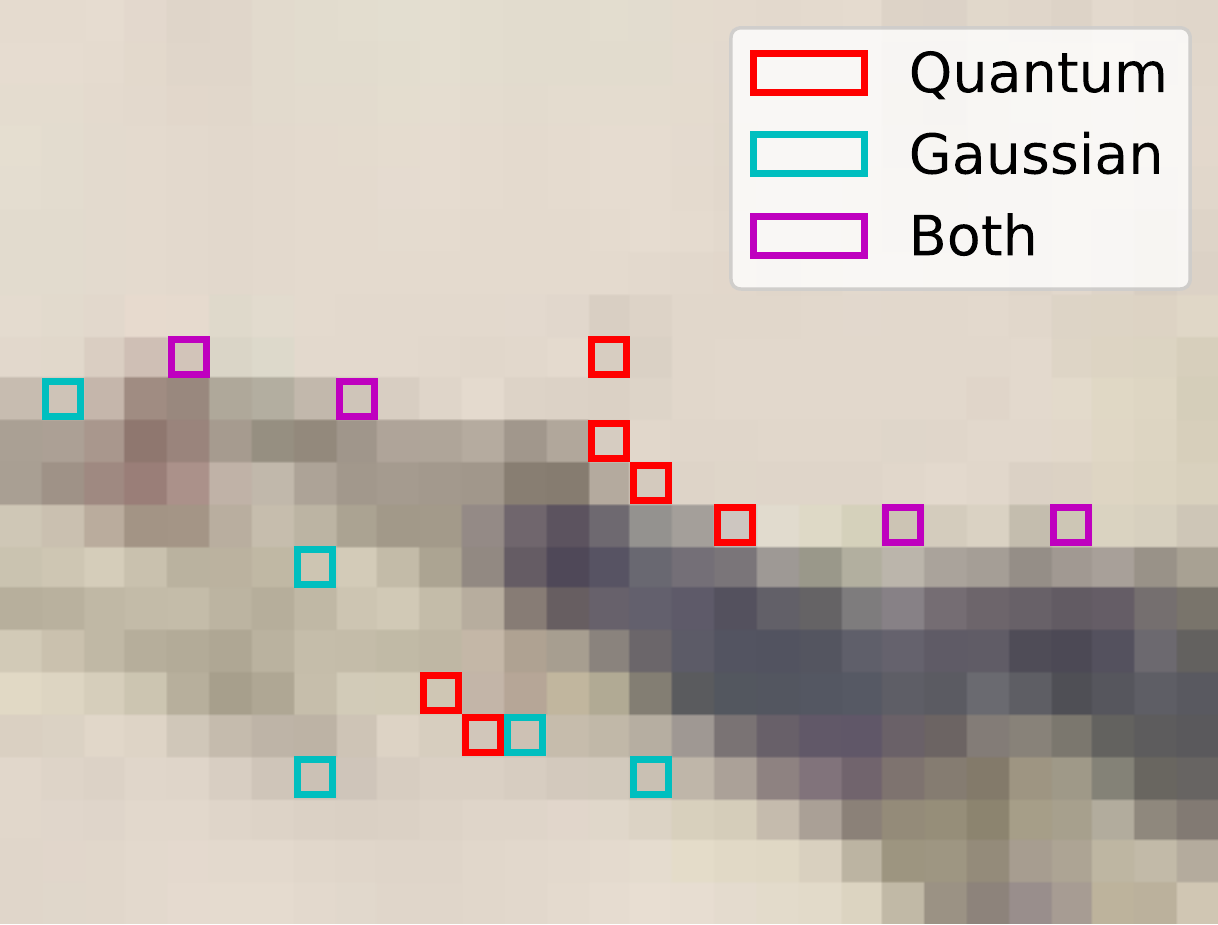}
		\caption{}
		\label{fig:keypoints:fourth}
	\end{subfigure}
	\caption{Extraction of 10 keypoints on $29\times22$ patches with KDC: Comparison between Gaussian and quantum kernel.}
	\label{fig:keypoints}
\end{figure*}

\begin{figure*}
	\centering
	\includegraphics[width=\textwidth]{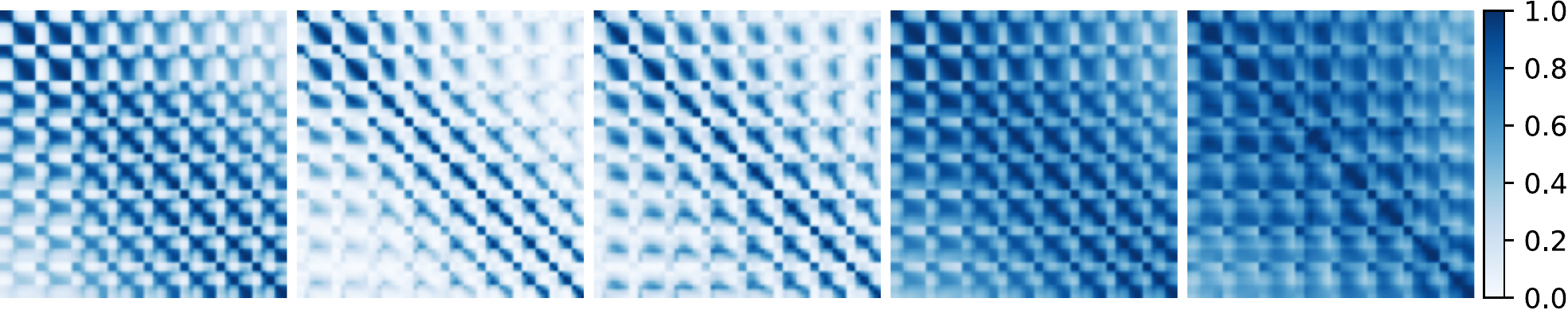}
	\caption{Comparison of kernel matrices computed on an exemplary $8\times8$ image patch shown in Fig.~\ref{fig:8x8:fifth}. Specifics on the used kernel from left to right: Gaussian kernel, quantum kernel computed with simulation, quantum kernel computed with real quantum hardware, quantum kernel computed with simulation with inputs being scaled by a factor $\lambda=0.5$, quantum kernel computed with real quantum hardware with inputs being scaled by a factor $\lambda=0.5$.}
	\label{fig:kernel_matrix}
\end{figure*}

\renewcommand{\arraystretch}{1.15}
\begin{table*}[t!]
	\centering
	\caption{Keypoint extraction on 10 $8\times 8$ patches: Comparison of energy values between the D-Wave quantum annealer Advantage System 5.1, simulated annealing and digital annealing.}
	\label{tab:8x8}
	\begin{tabular*}{\textwidth}{l @{\extracolsep{\fill}}cccccccccc}
\toprule
		& (a)    & (b)    & (c)    & (d)    & (e)    & (f)    & (g)    & (h)    & (i)    & (j)       \\
\midrule
		D-Wave Advantage System 5.1 & -4.714 & -4.797 & -4.737 & -4.690  & -4.751 & -4.343 & -4.294 & -4.404 & -4.442 & -4.402 \\
		Simulated Annealing & -4.742 & -4.818 & -4.768 & -4.748 & -4.780  & -4.790  & -4.726 & -4.782 & -4.809 & -4.793 \\
		Digital Annealing  & \bf{-4.750}  & \bf{-4.822} & \bf{-4.774} & \bf{-4.758} & \bf{-4.781} & \bf{-4.793} & \bf{-4.740}  & \bf{-4.788} & \bf{-4.813} & \bf{-4.797}\\
		\bottomrule
	\end{tabular*}
\end{table*}

\begin{figure*}[t!]
	\centering
	\begin{subfigure}[t]{0.19\textwidth}
		\centering
		\includegraphics[width=\textwidth]{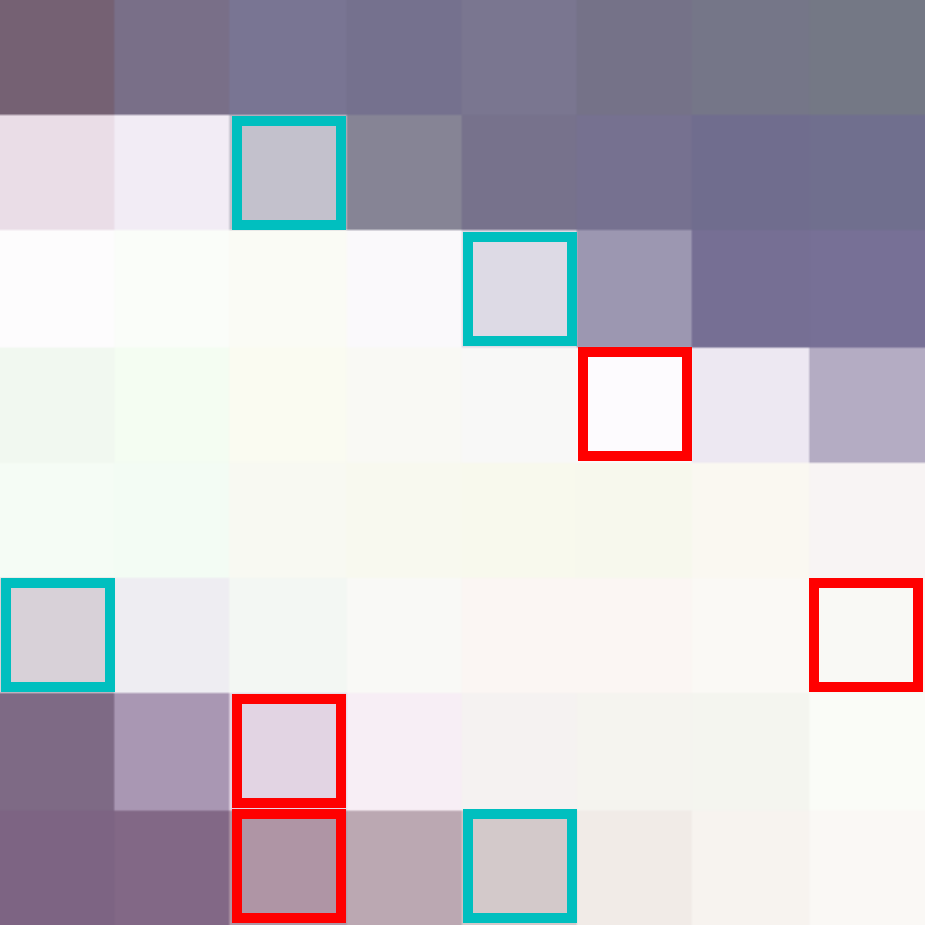}
		\caption{}
		\label{fig:8x8:first}
	\end{subfigure}
	\begin{subfigure}[t]{0.19\textwidth}
		\centering
		\includegraphics[width=\textwidth]{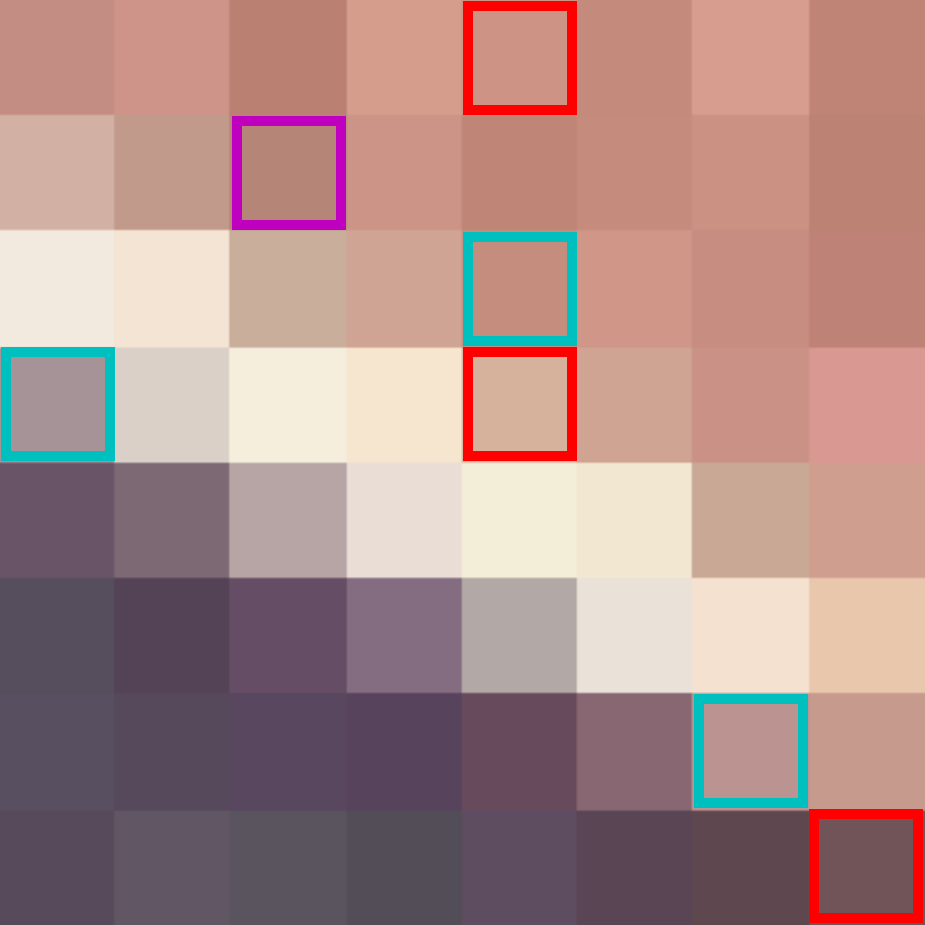}
		\caption{}
		\label{fig:8x8:second}
	\end{subfigure}
	\begin{subfigure}[t]{0.19\textwidth}
		\centering
		\includegraphics[width=\textwidth]{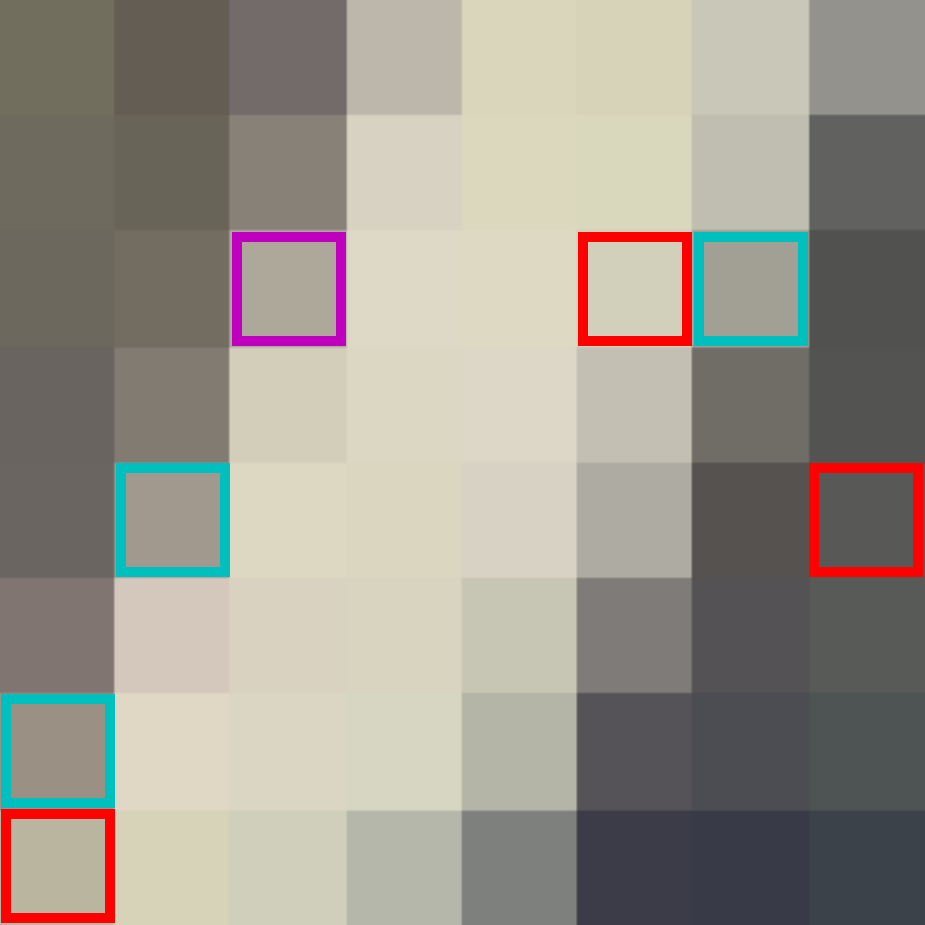}
		\caption{}
		\label{fig:8x8:third}
	\end{subfigure}
	\begin{subfigure}[t]{0.19\textwidth}
		\centering
		\includegraphics[width=\textwidth]{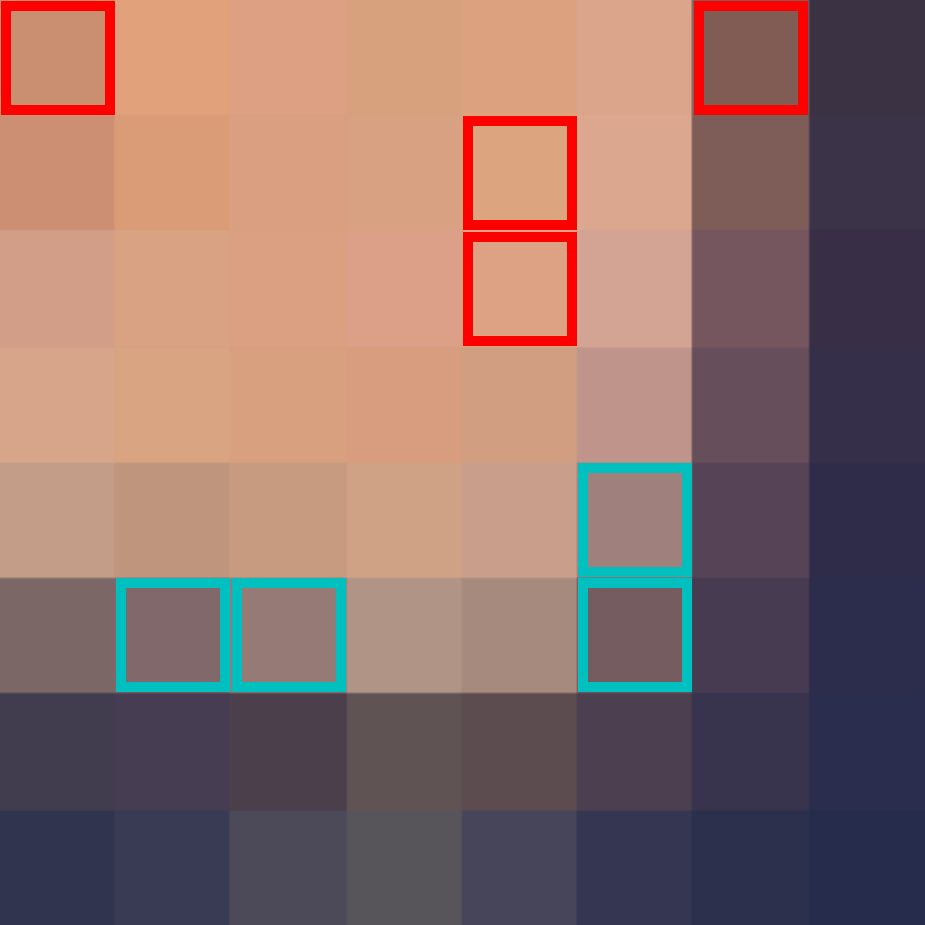}
		\caption{}
		\label{fig:8x8:fourth}
	\end{subfigure}
	\begin{subfigure}[t]{0.19\textwidth}
		\centering
		\includegraphics[width=\textwidth]{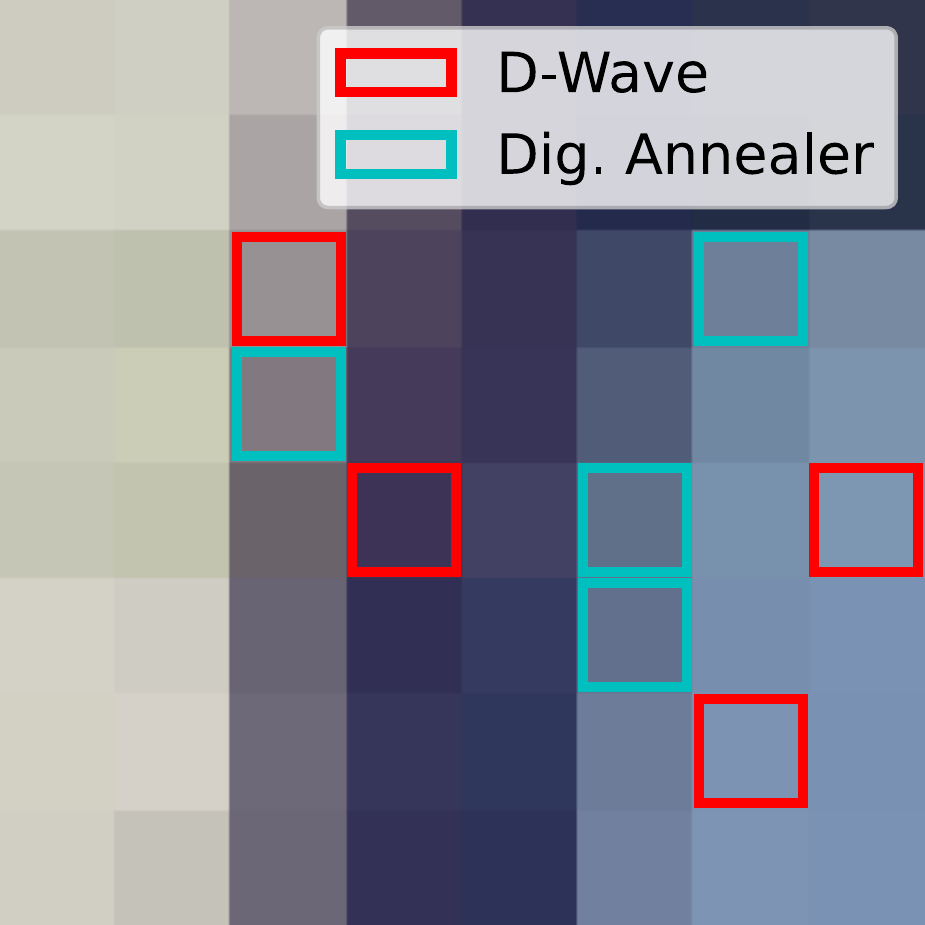}
		\caption{}
		\label{fig:8x8:fifth}
	\end{subfigure}
	\caption{Extraction of 10 keypoints on $8\times 8$ patches: Comparison of D-Wave quantum annealer and digital annealer.}
	\label{fig:8x8}
\end{figure*}

\subsection{Results}

Two exemplary results of the keypoint extraction pipeline can be found in Fig.~\ref{fig:keypoints_pipeline}. It is evident from every single subfigure that $k$-medoids clustering and KDC allocate the cluster centroids substantially different. While $k$-medoids clustering spreads its centroids equally distributed over the whole image patch, KDC is able to detect edges which is very useful for keypoint extraction. However, if such edges are represented by only a few pixels in the image patch (low density), KDC may not detect them. This is e.g. evident from the left and right neighbor patches of the highlighted patch in the top row of Fig.~\ref{fig:keypoints:third}. In such cases, the cluster centroids are driven towards the center of the patch, since most density is then captured in the position of the pixels---a proper re-weighting of pixel locations can hence be considered as a hyper parameter of the proposed method. In almost all cases, the digital annealer outperforms simulated annealing and tabu search. 

A comparison between the usage of a Gaussian kernel and a quantum kernel for KDC is depicted in Fig.~\ref{fig:keypoints}. Keypoints are extracted on four different $29\times 22$ image patches solving the KDC QUBO problem with the digital annealer. The quantum kernel is computed via Schrödinger wave-function / statevector simulation---we can see that KDC with a quantum kernel distributes its cluster centroids slightly different to the ones using a Gaussian kernel, while also capturing interesting pixel locations. Constructing a full quantum pipeline can hence be a viable approach.

Fig.~\ref{fig:kernel_matrix} depicts a comparison of kernel matrices of a Gaussian kernel with a quantum kernel computed from the $8\times8$ patch in Fig.~\ref{fig:8x8:fifth}. We here show the effects on the kernel matrices of scaling the inputs $\boldsymbol{x}\leftarrow \lambda\boldsymbol{x}$.
The quantum kernel matrices are computed for two different scales, $\lambda=1$ and $\lambda=0.5$.  We compare the results from the statevector simulation with the estimated kernel values using actual quantum hardware. One can see that the quantum kernels computed on actual hardware have a very similar structure to the simulated ones, while the scaling of the inputs substantially affects the \enquote{density} of the kernel matrix. 

In Fig.~\ref{fig:8x8} we compare the performance in solving the KDC QUBO problem with a quantum annealer and the digital annealer. 
Five $8\times8$ image patches are depicted with the corresponding extracted keypoints. Tab.~\ref{tab:8x8} shows the corresponding energy values of the best computed solution. It is clear that the digital annealer is finding better states in terms of objective function value.

Finally, exemplary solutions of the matching QUBO are depicted in Fig.~\ref{fig:matching}. For this, a sub-image with 10 keypoints is rotated by $20^{\circ}$ to obtain the same scene from a different view. The keypoints are then matched by solving the QUBO from Prop.~\ref{prop:matching}. In this case, the digital annealer needs a larger annealing time ($60s$) to find a good state, while tabu search can identify a good solution rather quickly. This shows that not only the QUBO dimension but especially the underlying energy landscape is of great importance for the performance of finding good states. We can see that setting $\alpha$ to a small value leads to finding only a few matches. However, the identified matches have the highest quality, i.e., the largest kernel values. The wrongly matched pair has a larger kernel value than the theoretically correct match, which can be ascribed to the feature representation of SIFT and is not an artefact of the underlying QUBO.

\begin{figure*}[h]
	\centering
	\begin{subfigure}[t]{0.45\textwidth}
		\centering
		\includegraphics[width=\textwidth]{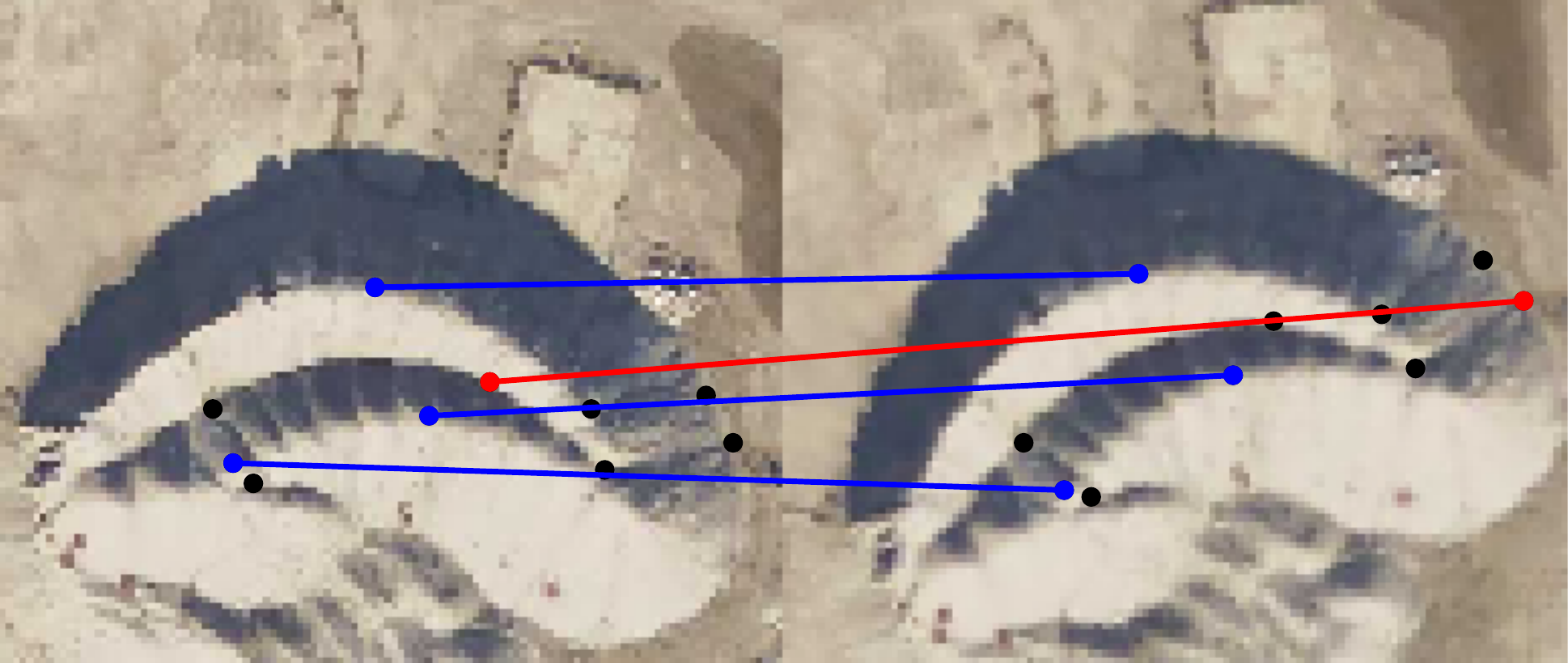}
		\caption{Emphasis on good matches ($\alpha=0.05$).}
		\label{fig:matching:first}
	\end{subfigure}
\hspace*{1cm}
	\begin{subfigure}[t]{0.45\textwidth}
		\centering
		\includegraphics[width=\textwidth]{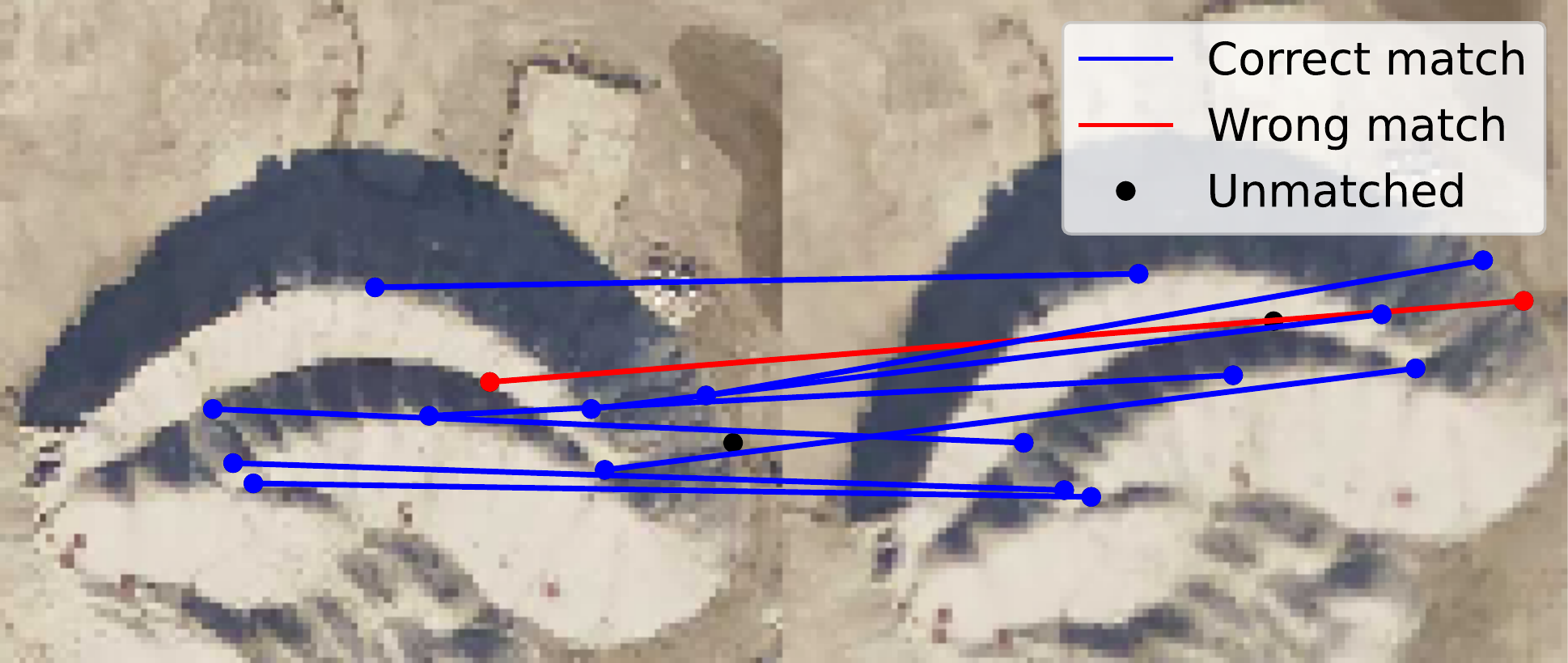}
		\caption{Emphasis on many matches ($\alpha=0.2$).}
		\label{fig:matching:second}
	\end{subfigure}
	\caption{Matching of 10 keypoints on a small image excerpt solving the matching QUBO in Prop.~\ref{prop:matching} with the digital annealer. The right image excerpt corresponds to the left one rotated clockwise by $20^{\circ}$.}
	\label{fig:matching}
\end{figure*}

\section{Conclusion}

In this work we approached the task of bundle adjustment via quantum machine learning. The feature extraction was re-interpreted as a clustering problem. QUBO formulations for the keypoint detection and matching problems have been derived. For the first time, we combined these QUBO problems with quantum kernels, which combines the adiabatic quantum computing paradigm with quantum gate computing. Experiments on actual quantum hardware and digital annealers show that the method delivers reasonable results. One cannot ultimately answer the question which hardware approach is best. Thus, investigating all available quantum computing resources is necessary at the time of writing.

Future work includes the optimization of hyper parameters (e.g., Lagrange multipliers) and investigating the \enquote{qubitization} of the full bundle adjustment task, e.g., formulating \eqref{eq:ba_objective} as a QUBO or quantum circuit. Moreover, quantum kernels can be used in the matching QUBO as well. Respecting the limitations of current quantum hardware, a lower dimensional feature descriptor than SIFT has to be chosen, e.g. PCA-SIFT \cite{pca_sift}. However, recent results lead to improvements for various hardware implementations of quantum computers \cite{Lescanne/etal/2020a}. In any case, our contributed methods open up opportunities for bundle adjustment and other computer vision tasks on the current and upcoming generations of quantum computing hardware. 

\balance


\end{document}